\documentclass[12pt,a4paper,fleqn]{article}

\usepackage[hidelinks]{hyperref}

\usepackage{agda}
\renewcommand{\AgdaFontStyle}[1]{\texttt{#1}}
\renewcommand{\AgdaKeywordFontStyle}[1]{\texttt{#1}}
\renewcommand{\AgdaBoundFontStyle}[1]{\emph{\texttt{#1}}}

\usepackage{a4wide}

\usepackage{amsmath,amssymb,amsthm}
\theoremstyle{plain}
\newtheorem{theorem}{Theorem}[section]
\newtheorem*{theorem*}{Theorem}

\newtheorem{lemma}[theorem]{Lemma}

\theoremstyle{definition}

\newtheorem*{note*}{Note}
\newtheorem*{notes*}{Notes}
\newtheorem*{null*}{}
\newtheorem{remark}[theorem]{Remark}
\newtheorem*{remark*}{Remark}

\usepackage[utf8]{inputenc}
\usepackage{newunicodechar}
\newunicodechar{ℕ}{\ensuremath{\mathbb{N}}}
\newunicodechar{→}{\ensuremath{\mathnormal\to}}
\newunicodechar{≡}{\ensuremath{\mathnormal\equiv}}
\newunicodechar{∀}{\ensuremath{\mathnormal\forall}}
\newunicodechar{∑}{\ensuremath{\mathnormal\textstyle\sum}}
\newunicodechar{⊔}{\ensuremath{\mathnormal\sqcup}}
\newunicodechar{′}{\ensuremath{\mathnormal{}'}}
\newunicodechar{∶}{\ensuremath{\mathnormal{}:}}
\newunicodechar{λ}{\ensuremath{\mathnormal{}\lambda}}
\newunicodechar{₂}{\ensuremath{\mathnormal{}_2}}
\newunicodechar{₃}{\ensuremath{\mathnormal{}_3}}
\newunicodechar{Σ}{\ensuremath{\mathnormal\Sigma}}
\newunicodechar{≐}{\ensuremath{\mathnormal\doteq}}

\usepackage[no-weekday]{eukdate}

\usepackage[pdftex]{graphics}

\usepackage{mathtools}

\usepackage{natbib}

\usepackage{stmaryrd}

\usepackage{url}

\usepackage[tt=false, type1=true]{libertine}
\usepackage[libertine]{newtxmath}
\usepackage[sans]{dsfont}
\usepackage[mathscr]{eucal}

\newcommand{\aarg}[1]{\AgdaArgument{#1}}
\newcommand{\afun}[1]{\AgdaFunction{#1}}
\newcommand{\atyp}[1]{\AgdaPrimitiveType{#1}}
\newcommand{\eq}{\mathrel{\afun{≡}}}
\newcommand{\fun}{\mathop{\AgdaSymbol{→}}}
\newcommand{\heq}{\mathrel{\afun{≡≡}}}

\title{Typal Heterogeneous Equality Types}

\author{Andrew M.Pitts}

\date{University of Cambridge, UK} 
 
\begin{document}

\maketitle

\begin{abstract}
  The usual homogeneous form of equality type in Martin-L\"of Type
  Theory contains identifications between elements of the same
  type. By contrast, the heterogeneous form of equality contains
  identifications between elements of possibly different types. This
  paper introduces a simple set of axioms for such types. The axioms
  are equivalent to the combination of systematic elimination rules
  for both forms of equality, albeit with typal (also known as
  ``propositional'') computation properties, together with Streicher's
  Axiom~K, or equivalently, the principle of uniqueness of identity
  proofs.
\end{abstract}

\section{Introduction}
\label{sec:int}

Equality types in the intensional version of Martin-L\"of Type
Theory~\citep[see for example][Section~8.1]{NordstromB:progmlt} are
traditionally formulated in terms of an introduction rule
(reflexivity) together with a rule for eliminating proofs of equality
and a rule describing how elimination computes when it meets a
reflexivity proof. Some recent work
\citep{CoquandT:modttc,CoquandT:cubttc} on models of Homotopy Type
Theory~\citep{HoTT} uses a formulation of equality types that differs
from this in two respects. First, the elimination operation is
replaced by the combination of a simple operation for transporting
elements along proofs of equality, together with an axiom asserting
contractibility of singleton types. Secondly, the analogue of the
computation rule for the eliminator, namely that transporting along a
reflexivity proof does nothing, is weakened from a judgemental
equality to the existence of an element of the corresponding equality
type; see~\cite{CoquandT:equdtt} and Figure~2 in
\citep{CoquandT:modttc}.  This formulation is sometimes called a
``propositional'' equality type~\citep{VanDenBergB:patcpi}, but here I
will follow \citet[Section~1.6]{ShulmanM:brofpt} for the reasons given
there and refer to \emph{typal} equality types.  Although these
changes to the formulation of equality types affect computation, it
seems that they do not change what is provable (see
\cite{BoulierS:weattr} for example) and they make it easier to
construct models. Furthermore, they can lead to surprising
simplifications. For example, Lumsdaine~[private communication] has
observed that the computation rule is superfluous (for elimination,
but the observation also holds for transport): if a proto-identity
type has a transport operation lacking its typal computation property,
then the operation can be corrected to a new one that does have the
computation property (see Lemma~\ref{lem:verlt} and the Appendix).

\begin{figure}

\centering

\begin{code}[hide]%
\>[0]\AgdaSymbol{\{-\#}\AgdaSpace{}%
\AgdaKeyword{OPTIONS}\AgdaSpace{}%
\AgdaPragma{--with-K}\AgdaSpace{}%
\AgdaSymbol{\#-\}}\<%
\\
\\[\AgdaEmptyExtraSkip]%
\>[0]\AgdaKeyword{module}\AgdaSpace{}%
\AgdaModule{typhet}\AgdaSpace{}%
\AgdaKeyword{where}\<%
\\
\>[0]\<%
\\
\>[0]\AgdaKeyword{open}\AgdaSpace{}%
\AgdaKeyword{import}\AgdaSpace{}%
\AgdaModule{Agda.Primitive}\<%
\\
\\[\AgdaEmptyExtraSkip]%
\>[0]\AgdaKeyword{module}\AgdaSpace{}%
\AgdaModule{Axioms}\AgdaSpace{}%
\AgdaKeyword{where}\<%
\\
\>[0][@{}l@{\AgdaIndent{0}}]%
\>[2]\AgdaKeyword{infix}%
\>[9]\AgdaNumber{1}\AgdaSpace{}%
\AgdaPostulate{∑}\AgdaSpace{}%
\AgdaOperator{\AgdaFunction{proof\AgdaUnderscore{}}}\<%
\\
\>[2]\AgdaKeyword{infixr}\AgdaSpace{}%
\AgdaNumber{2}\AgdaSpace{}%
\AgdaOperator{\AgdaFunction{\AgdaUnderscore{}≡≡[\AgdaUnderscore{}]\AgdaUnderscore{}}}\<%
\\
\>[2]\AgdaKeyword{infix}%
\>[9]\AgdaNumber{3}\AgdaSpace{}%
\AgdaOperator{\AgdaFunction{\AgdaUnderscore{}qed}}\<%
\\
\>[2]\AgdaKeyword{infix}%
\>[9]\AgdaNumber{4}\AgdaSpace{}%
\AgdaOperator{\AgdaPostulate{\AgdaUnderscore{}≡≡\AgdaUnderscore{}}}\AgdaSpace{}%
\AgdaOperator{\AgdaFunction{\AgdaUnderscore{}≡\AgdaUnderscore{}}}\<%
\end{code}

Axioms for typal heterogeneous equality satisfying Axiom~K

\AgdaTarget{≡≡}
\AgdaTarget{≡}
\AgdaTarget{rfl}
\AgdaTarget{ctr}
\AgdaTarget{eqt}
\AgdaTarget{tpt}
\AgdaTarget{∑}
\AgdaTarget{\_,\_}
\AgdaTarget{fst}
\AgdaTarget{snd}
\AgdaTarget{fpr}
\AgdaTarget{spr}
\AgdaTarget{eta}
\AgdaNoSpaceAroundCode{}

\begin{AgdaAlign}
\begin{code}  %
\>[2]\AgdaKeyword{postulate}\<%
\\
\>[2][@{}l@{\AgdaIndent{0}}]%
\>[4]\AgdaOperator{\AgdaPostulate{\AgdaUnderscore{}≡≡\AgdaUnderscore{}}}\AgdaSpace{}%
\AgdaSymbol{:}\AgdaSpace{}%
\AgdaSymbol{∀\{}\AgdaBound{l}\AgdaSymbol{\}\{}\AgdaBound{A}\AgdaSpace{}%
\AgdaBound{B}\AgdaSpace{}%
\AgdaSymbol{:}\AgdaSpace{}%
\AgdaPrimitiveType{Set}\AgdaSpace{}%
\AgdaBound{l}\AgdaSymbol{\}}\AgdaSpace{}%
\AgdaSymbol{→}\AgdaSpace{}%
\AgdaBound{A}\AgdaSpace{}%
\AgdaSymbol{→}\AgdaSpace{}%
\AgdaBound{B}\AgdaSpace{}%
\AgdaSymbol{→}\AgdaSpace{}%
\AgdaPrimitiveType{Set}\AgdaSpace{}%
\AgdaBound{l}\<%
\\
\\[\AgdaEmptyExtraSkip]%
\>[2]\AgdaComment{-- the derived homogeneous equality}\<%
\\
\>[2]\AgdaOperator{\AgdaFunction{\AgdaUnderscore{}≡\AgdaUnderscore{}}}%
\>[9]\AgdaSymbol{:}\AgdaSpace{}%
\AgdaSymbol{∀\{}\AgdaBound{l}\AgdaSymbol{\}\{}\AgdaBound{A}\AgdaSpace{}%
\AgdaSymbol{:}\AgdaSpace{}%
\AgdaPrimitiveType{Set}\AgdaSpace{}%
\AgdaBound{l}\AgdaSymbol{\}}\AgdaSpace{}%
\AgdaSymbol{→}\AgdaSpace{}%
\AgdaBound{A}\AgdaSpace{}%
\AgdaSymbol{→}\AgdaSpace{}%
\AgdaBound{A}\AgdaSpace{}%
\AgdaSymbol{→}\AgdaSpace{}%
\AgdaPrimitiveType{Set}\AgdaSpace{}%
\AgdaBound{l}\<%
\\
\>[2]\AgdaBound{x}\AgdaSpace{}%
\AgdaOperator{\AgdaFunction{≡}}\AgdaSpace{}%
\AgdaBound{y}%
\>[9]\AgdaSymbol{=}\AgdaSpace{}%
\AgdaBound{x}\AgdaSpace{}%
\AgdaOperator{\AgdaPostulate{≡≡}}\AgdaSpace{}%
\AgdaBound{y}\<%
\\
\\[\AgdaEmptyExtraSkip]%
\>[2]\AgdaKeyword{postulate}\<%
\\
\>[2][@{}l@{\AgdaIndent{0}}]%
\>[4]\AgdaPostulate{rfl}%
\>[9]\AgdaSymbol{:}\AgdaSpace{}%
\AgdaSymbol{∀\{}\AgdaBound{l}\AgdaSymbol{\}\{}\AgdaBound{A}\AgdaSpace{}%
\AgdaSymbol{:}\AgdaSpace{}%
\AgdaPrimitiveType{Set}\AgdaSpace{}%
\AgdaBound{l}\AgdaSymbol{\}}\AgdaSpace{}%
\AgdaSymbol{(}\AgdaBound{x}\AgdaSpace{}%
\AgdaSymbol{:}\AgdaSpace{}%
\AgdaBound{A}\AgdaSymbol{)}\AgdaSpace{}%
\AgdaSymbol{→}\AgdaSpace{}%
\AgdaBound{x}\AgdaSpace{}%
\AgdaOperator{\AgdaFunction{≡}}\AgdaSpace{}%
\AgdaBound{x}\<%
\\
\>[4]\AgdaPostulate{ctr}%
\>[9]\AgdaSymbol{:}\AgdaSpace{}%
\AgdaSymbol{∀\{}\AgdaBound{l}\AgdaSymbol{\}\{}\AgdaBound{A}\AgdaSpace{}%
\AgdaBound{B}\AgdaSpace{}%
\AgdaSymbol{:}\AgdaSpace{}%
\AgdaPrimitiveType{Set}\AgdaSpace{}%
\AgdaBound{l}\AgdaSymbol{\}\{}\AgdaBound{x}\AgdaSpace{}%
\AgdaSymbol{:}\AgdaSpace{}%
\AgdaBound{A}\AgdaSymbol{\}\{}\AgdaBound{y}\AgdaSpace{}%
\AgdaSymbol{:}\AgdaSpace{}%
\AgdaBound{B}\AgdaSymbol{\}(}\AgdaBound{e}\AgdaSpace{}%
\AgdaSymbol{:}\AgdaSpace{}%
\AgdaBound{x}\AgdaSpace{}%
\AgdaOperator{\AgdaPostulate{≡≡}}\AgdaSpace{}%
\AgdaBound{y}\AgdaSymbol{)}\AgdaSpace{}%
\AgdaSymbol{→}\AgdaSpace{}%
\AgdaPostulate{rfl}\AgdaSpace{}%
\AgdaBound{x}\AgdaSpace{}%
\AgdaOperator{\AgdaPostulate{≡≡}}\AgdaSpace{}%
\AgdaBound{e}\<%
\\
\>[4]\AgdaPostulate{eqt}%
\>[9]\AgdaSymbol{:}\AgdaSpace{}%
\AgdaSymbol{∀\{}\AgdaBound{l}\AgdaSymbol{\}\{}\AgdaBound{A}\AgdaSpace{}%
\AgdaBound{B}\AgdaSpace{}%
\AgdaSymbol{:}\AgdaSpace{}%
\AgdaPrimitiveType{Set}\AgdaSpace{}%
\AgdaBound{l}\AgdaSymbol{\}\{}\AgdaBound{x}\AgdaSpace{}%
\AgdaSymbol{:}\AgdaSpace{}%
\AgdaBound{A}\AgdaSymbol{\}\{}\AgdaBound{y}\AgdaSpace{}%
\AgdaSymbol{:}\AgdaSpace{}%
\AgdaBound{B}\AgdaSymbol{\}}\AgdaSpace{}%
\AgdaSymbol{→}\AgdaSpace{}%
\AgdaBound{x}\AgdaSpace{}%
\AgdaOperator{\AgdaPostulate{≡≡}}\AgdaSpace{}%
\AgdaBound{y}\AgdaSpace{}%
\AgdaSymbol{→}\AgdaSpace{}%
\AgdaBound{A}\AgdaSpace{}%
\AgdaOperator{\AgdaFunction{≡}}\AgdaSpace{}%
\AgdaBound{B}\<%
\\
\>[4]\AgdaPostulate{tpt}%
\>[9]\AgdaSymbol{:}%
\>[91I]\AgdaSymbol{∀\{}\AgdaBound{l}\AgdaSpace{}%
\AgdaBound{m}\AgdaSpace{}%
\AgdaBound{n}\AgdaSymbol{\}\{}\AgdaBound{A}\AgdaSpace{}%
\AgdaSymbol{:}\AgdaSpace{}%
\AgdaPrimitiveType{Set}\AgdaSpace{}%
\AgdaBound{l}\AgdaSymbol{\}\{}\AgdaBound{B}\AgdaSpace{}%
\AgdaSymbol{:}\AgdaSpace{}%
\AgdaBound{A}\AgdaSpace{}%
\AgdaSymbol{→}\AgdaSpace{}%
\AgdaPrimitiveType{Set}\AgdaSpace{}%
\AgdaBound{m}\AgdaSymbol{\}(}\AgdaBound{C}\AgdaSpace{}%
\AgdaSymbol{:}\AgdaSpace{}%
\AgdaSymbol{(}\AgdaBound{x}\AgdaSpace{}%
\AgdaSymbol{:}\AgdaSpace{}%
\AgdaBound{A}\AgdaSymbol{)}\AgdaSpace{}%
\AgdaSymbol{→}\AgdaSpace{}%
\AgdaBound{B}\AgdaSpace{}%
\AgdaBound{x}\AgdaSpace{}%
\AgdaSymbol{→}\AgdaSpace{}%
\AgdaPrimitiveType{Set}\AgdaSpace{}%
\AgdaBound{n}\AgdaSymbol{)}\<%
\\
\>[.][@{}l@{}]\<[91I]%
\>[11]\AgdaSymbol{\{}\AgdaBound{x}\AgdaSpace{}%
\AgdaBound{x′}\AgdaSpace{}%
\AgdaSymbol{:}\AgdaSpace{}%
\AgdaBound{A}\AgdaSymbol{\}\{}\AgdaBound{y}\AgdaSpace{}%
\AgdaSymbol{:}\AgdaSpace{}%
\AgdaBound{B}\AgdaSpace{}%
\AgdaBound{x}\AgdaSymbol{\}\{}\AgdaBound{y′}\AgdaSpace{}%
\AgdaSymbol{:}\AgdaSpace{}%
\AgdaBound{B}\AgdaSpace{}%
\AgdaBound{x′}\AgdaSymbol{\}}\AgdaSpace{}%
\AgdaSymbol{→}\AgdaSpace{}%
\AgdaBound{x}\AgdaSpace{}%
\AgdaOperator{\AgdaFunction{≡}}\AgdaSpace{}%
\AgdaBound{x′}\AgdaSpace{}%
\AgdaSymbol{→}\AgdaSpace{}%
\AgdaBound{y}\AgdaSpace{}%
\AgdaOperator{\AgdaPostulate{≡≡}}\AgdaSpace{}%
\AgdaBound{y′}\AgdaSpace{}%
\AgdaSymbol{→}\AgdaSpace{}%
\AgdaBound{C}\AgdaSpace{}%
\AgdaBound{x}\AgdaSpace{}%
\AgdaBound{y}\AgdaSpace{}%
\AgdaSymbol{→}\AgdaSpace{}%
\AgdaBound{C}\AgdaSpace{}%
\AgdaBound{x′}\AgdaSpace{}%
\AgdaBound{y′}\<%
\end{code}
           
\bigskip

Axioms for $\Sigma$-types with surjective pairing

\begin{code} %
\>[2]\AgdaKeyword{postulate}\<%
\\
\>[2][@{}l@{\AgdaIndent{0}}]%
\>[4]\AgdaPostulate{∑}%
\>[12]\AgdaSymbol{:}\AgdaSpace{}%
\AgdaSymbol{∀\{}\AgdaBound{l}\AgdaSpace{}%
\AgdaBound{m}\AgdaSymbol{\}(}\AgdaBound{A}\AgdaSpace{}%
\AgdaSymbol{:}\AgdaSpace{}%
\AgdaPrimitiveType{Set}\AgdaSpace{}%
\AgdaBound{l}\AgdaSymbol{)(}\AgdaBound{B}\AgdaSpace{}%
\AgdaSymbol{:}\AgdaSpace{}%
\AgdaBound{A}\AgdaSpace{}%
\AgdaSymbol{→}\AgdaSpace{}%
\AgdaPrimitiveType{Set}\AgdaSpace{}%
\AgdaBound{m}\AgdaSymbol{)}\AgdaSpace{}%
\AgdaSymbol{→}\AgdaSpace{}%
\AgdaPrimitiveType{Set}\AgdaSpace{}%
\AgdaSymbol{(}\AgdaBound{l}\AgdaSpace{}%
\AgdaOperator{\AgdaPrimitive{⊔}}\AgdaSpace{}%
\AgdaBound{m}\AgdaSymbol{)}\<%
\\
\\[\AgdaEmptyExtraSkip]%
\>[2]\AgdaKeyword{module}\AgdaSpace{}%
\AgdaModule{\AgdaUnderscore{}}\AgdaSpace{}%
\AgdaSymbol{\{}\AgdaBound{l}\AgdaSpace{}%
\AgdaBound{m}\AgdaSymbol{\}\{}\AgdaBound{A}\AgdaSpace{}%
\AgdaSymbol{:}\AgdaSpace{}%
\AgdaPrimitiveType{Set}\AgdaSpace{}%
\AgdaBound{l}\AgdaSymbol{\}\{}\AgdaBound{B}\AgdaSpace{}%
\AgdaSymbol{:}\AgdaSpace{}%
\AgdaBound{A}\AgdaSpace{}%
\AgdaSymbol{→}\AgdaSpace{}%
\AgdaPrimitiveType{Set}\AgdaSpace{}%
\AgdaBound{m}\AgdaSymbol{\}}\AgdaSpace{}%
\AgdaKeyword{where}\<%
\end{code}
\begin{code}[hide]%
\>[2][@{}l@{\AgdaIndent{1}}]%
\>[4]\AgdaKeyword{infix}\AgdaSpace{}%
\AgdaNumber{3}\AgdaSpace{}%
\AgdaOperator{\AgdaPostulate{\AgdaUnderscore{},\AgdaUnderscore{}}}\<%
\end{code}
\begin{code} %
\>[4]\AgdaKeyword{postulate}\<%
\\
\>[4][@{}l@{\AgdaIndent{0}}]%
\>[6]\AgdaOperator{\AgdaPostulate{\AgdaUnderscore{},\AgdaUnderscore{}}}%
\>[11]\AgdaSymbol{:}\AgdaSpace{}%
\AgdaSymbol{(}\AgdaBound{x}\AgdaSpace{}%
\AgdaSymbol{:}\AgdaSpace{}%
\AgdaBound{A}\AgdaSymbol{)}\AgdaSpace{}%
\AgdaSymbol{→}\AgdaSpace{}%
\AgdaBound{B}\AgdaSpace{}%
\AgdaBound{x}\AgdaSpace{}%
\AgdaSymbol{→}\AgdaSpace{}%
\AgdaPostulate{∑}\AgdaSpace{}%
\AgdaBound{A}\AgdaSpace{}%
\AgdaBound{B}\<%
\\
\>[6]\AgdaPostulate{fst}%
\>[11]\AgdaSymbol{:}\AgdaSpace{}%
\AgdaPostulate{∑}\AgdaSpace{}%
\AgdaBound{A}\AgdaSpace{}%
\AgdaBound{B}\AgdaSpace{}%
\AgdaSymbol{→}\AgdaSpace{}%
\AgdaBound{A}\<%
\\
\>[6]\AgdaPostulate{snd}%
\>[11]\AgdaSymbol{:}\AgdaSpace{}%
\AgdaSymbol{(}\AgdaBound{z}\AgdaSpace{}%
\AgdaSymbol{:}\AgdaSpace{}%
\AgdaPostulate{∑}\AgdaSpace{}%
\AgdaBound{A}\AgdaSpace{}%
\AgdaBound{B}\AgdaSymbol{)}\AgdaSpace{}%
\AgdaSymbol{→}\AgdaSpace{}%
\AgdaBound{B}\AgdaSpace{}%
\AgdaSymbol{(}\AgdaPostulate{fst}\AgdaSpace{}%
\AgdaBound{z}\AgdaSymbol{)}\<%
\\
\>[6]\AgdaPostulate{fpr}%
\>[11]\AgdaSymbol{:}\AgdaSpace{}%
\AgdaSymbol{(}\AgdaBound{x}\AgdaSpace{}%
\AgdaSymbol{:}\AgdaSpace{}%
\AgdaBound{A}\AgdaSymbol{)(}\AgdaBound{y}\AgdaSpace{}%
\AgdaSymbol{:}\AgdaSpace{}%
\AgdaBound{B}\AgdaSpace{}%
\AgdaBound{x}\AgdaSymbol{)}\AgdaSpace{}%
\AgdaSymbol{→}\AgdaSpace{}%
\AgdaPostulate{fst}\AgdaSymbol{(}\AgdaBound{x}\AgdaSpace{}%
\AgdaOperator{\AgdaPostulate{,}}\AgdaSpace{}%
\AgdaBound{y}\AgdaSymbol{)}\AgdaSpace{}%
\AgdaOperator{\AgdaFunction{≡}}\AgdaSpace{}%
\AgdaBound{x}\<%
\\
\>[6]\AgdaPostulate{spr}%
\>[11]\AgdaSymbol{:}\AgdaSpace{}%
\AgdaSymbol{(}\AgdaBound{x}\AgdaSpace{}%
\AgdaSymbol{:}\AgdaSpace{}%
\AgdaBound{A}\AgdaSymbol{)(}\AgdaBound{y}\AgdaSpace{}%
\AgdaSymbol{:}\AgdaSpace{}%
\AgdaBound{B}\AgdaSpace{}%
\AgdaBound{x}\AgdaSymbol{)}\AgdaSpace{}%
\AgdaSymbol{→}\AgdaSpace{}%
\AgdaPostulate{snd}\AgdaSymbol{(}\AgdaBound{x}\AgdaSpace{}%
\AgdaOperator{\AgdaPostulate{,}}\AgdaSpace{}%
\AgdaBound{y}\AgdaSymbol{)}\AgdaSpace{}%
\AgdaOperator{\AgdaPostulate{≡≡}}\AgdaSpace{}%
\AgdaBound{y}\<%
\\
\>[6]\AgdaPostulate{eta}%
\>[11]\AgdaSymbol{:}\AgdaSpace{}%
\AgdaSymbol{(}\AgdaBound{z}\AgdaSpace{}%
\AgdaSymbol{:}\AgdaSpace{}%
\AgdaPostulate{∑}\AgdaSpace{}%
\AgdaBound{A}\AgdaSpace{}%
\AgdaBound{B}\AgdaSymbol{)}\AgdaSpace{}%
\AgdaSymbol{→}\AgdaSpace{}%
\AgdaSymbol{(}\AgdaPostulate{fst}\AgdaSpace{}%
\AgdaBound{z}\AgdaSpace{}%
\AgdaOperator{\AgdaPostulate{,}}\AgdaSpace{}%
\AgdaPostulate{snd}\AgdaSpace{}%
\AgdaBound{z}\AgdaSymbol{)}\AgdaSpace{}%
\AgdaOperator{\AgdaFunction{≡}}\AgdaSpace{}%
\AgdaBound{z}\<%
\\
\\[\AgdaEmptyExtraSkip]%
\>[2]\AgdaComment{-- concrete syntax for ∑-types}\<%
\\
\>[2]\AgdaKeyword{syntax}\AgdaSpace{}%
\AgdaPostulate{∑}\AgdaSpace{}%
\AgdaBound{A}\AgdaSpace{}%
\AgdaSymbol{(λ}\AgdaSpace{}%
\AgdaBound{x}\AgdaSpace{}%
\AgdaSymbol{→}\AgdaSpace{}%
\AgdaBound{B}\AgdaSymbol{)}\AgdaSpace{}%
\AgdaSymbol{=}\AgdaSpace{}%
\AgdaPostulate{∑}\AgdaSpace{}%
\AgdaBound{x}\AgdaSpace{}%
\AgdaPostulate{∶}\AgdaSpace{}%
\AgdaBound{A}\AgdaSpace{}%
\AgdaPostulate{,}\AgdaSpace{}%
\AgdaBound{B}\<%
\end{code}
\end{AgdaAlign}
\AgdaSpaceAroundCode{}

\caption{The Axioms}
\label{fig:axihtpe}
\end{figure}

The above remarks apply to the usual, homogeneous notion of equality
in which elements of the same type are compared. The purpose of this
paper is to give an analogous treatment of \emph{heterogeneous}
equality~\citep{McBrideC:deptfp,AltenkirchT:obsen} in the presence of
$\Sigma$-types and the Axiom~K of
\citet[Section~1.2]{SteicherT:invitt}.  Since Axiom~K is not
compatible with the Univalence Principle of Homotopy Type
Theory~\citep[Example~3.1.9]{HoTT}, the focus here is on the simpler
(but still useful!) world of zero-dimensional type theory.  We will
see that the axioms in Fig.~\ref{fig:axihtpe} capture homogeneous and
heterogeneous equality satisfying their usual dependent elimination
and (typal) computation properties and Axiom~K, and $\Sigma$-types
with their usual dependent elimination and (typal) computation
properties.  It seems necessary to include $\Sigma$-types in order to
get Lumsdaine's result mentioned above (see Remark~\ref{rem:rolst});
the axioms we give for such types are standard, except that the
equality property of dependent second projection (\afun{spr}) is
simplified by the use of heterogeneous rather than homogeneous
equality. The axioms in the figure are pleasingly simple compared to
the usual formulation in terms of elimination and computation
properties, and may aid finding new models of heterogeneous equality
types.

The implementation of intensional Martin-L\"of Type Theory provided by
Agda~2.6~\citep{Agda} is used to state the axioms and develop their
properties. More precisely, we just make use of Agda's implementation
of a countably infinite, non-cumulative hierarchy of universes
$\atyp{Set}\,\AgdaBound{l}$, where $\AgdaBound{l}$ ranges
over a type $\atyp{Level}$ of universe levels whose closed normal
forms are in bijection with the natural numbers. The universes are
closed under dependent function types (written in Agda as
$(\aarg{x}:\aarg{A})\fun\aarg{B}\,$) and inductive types.  The use of
a whole hierarchy of universes is necessary; for example, the function
$\afun{eqt}$ in Fig.~\ref{fig:axihtpe} takes a heterogeneous equality
type $\aarg{x}\heq\aarg{y}$ in universe $\afun{Set}\,\aarg{l}$ and
produces a homogeneous one $\aarg{A}\eq\aarg{B}$ in the universe one
level up, which is denoted $\afun{Set}(\afun{lsuc}\,\aarg{l})$ in
Agda. We also use Agda's notation for infix and for implicit
arguments. For example, the function $\_{\heq}\_$ in
Fig.~\ref{fig:axihtpe} takes five arguments, the first three of which
are implicit and the last two of which are infix. In particular,
Agda's ability to infer the values of implicit arguments (or of
unspecified explicit arguments, which are denoted by an underscore,
$\,\_\,$) is used quite aggressively in what follows, in order to be
able to see the wood from the trees.

Although the code in this paper has been checked by Agda, some parts
of it that are not essential for understanding the development have
been hidden; the complete (non-literate) Agda code can be found
at~[\url{https://doi.org/10.17863/CAM.47902}]. 

\section{The axioms and their properties}

Figure~\ref{fig:axihtpe} postulates a family of types $\_{\heq}\_$ in
all universes, together with some operations on them that together
capture a typal version of \emph{heterogeneous}
equality. Heterogeneous equality types were introduced by
\citet[Section~5.1.3]{McBrideC:deptfp} under the name of ``John Major
equality''. Unlike ordinary, homogeneous equality types, such a type
$\aarg{x}\heq\aarg{y}$ relates elements $\aarg{x}$ and $\aarg{y}$ of
possibly different types, $\aarg{A}$ and $\aarg{B}$ say. The intention
is that elements of type $\aarg{x}\heq\aarg{y}$ denote proofs that not
only are $\aarg{x}$ and $\aarg{y}$ equal, but so also are their types
$\aarg{A}$ and $\aarg{B}$. The figure defines homogeneous equality
$\_{\eq}\_$ as the special case of $\_{\heq}\_$ when the types of the
two arguments are already known to be the same.  Axiom~$\afun{rfl}$
says that $\eq$ is reflexive. Axiom~$\afun{ctr}$ is a heterogeneous
version of the contractibility property of singleton types
(cf.~$\mathsf{center}$ in Figure~2 of
\cite{CoquandT:modttc}). Axiom~$\afun{eqt}$ says that heterogeneously
equal things have (homogeneously) equal types. Axiom~$\afun{tpt}$ is a
form of the transport property of equality (cf.~$\mathsf{T}$ in
Figure~2 of \cite{CoquandT:modttc}) involving both homogeneous and
heterogeneous equalities. Finally, $\afun{∑}$, $\_\afun{,}\_$,
$\afun{fst}$, $\afun{snd}$, $\afun{fpr}$, $\afun{spr}$ and
$\afun{eta}$ axiomatize dependent product types satisfying surjective
pairing.

\begin{figure}
  \AgdaTarget{symm}
  \AgdaTarget{proof\_}
  \AgdaTarget{\_≡≡[\_]\_}
  \AgdaTarget{\_qed}
  \AgdaTarget{cong}
  \AgdaTarget{cong₂}
  \AgdaTarget{cong₃}
\begin{code}%
\>[2]\AgdaFunction{symm}\AgdaSpace{}%
\AgdaSymbol{:}\AgdaSpace{}%
\AgdaSymbol{∀\{}\AgdaBound{l}\AgdaSymbol{\}\{}\AgdaBound{A}\AgdaSpace{}%
\AgdaBound{B}\AgdaSpace{}%
\AgdaSymbol{:}\AgdaSpace{}%
\AgdaPrimitiveType{Set}\AgdaSpace{}%
\AgdaBound{l}\AgdaSymbol{\}\{}\AgdaBound{x}\AgdaSpace{}%
\AgdaSymbol{:}\AgdaSpace{}%
\AgdaBound{A}\AgdaSymbol{\}\{}\AgdaBound{y}\AgdaSpace{}%
\AgdaSymbol{:}\AgdaSpace{}%
\AgdaBound{B}\AgdaSymbol{\}}\AgdaSpace{}%
\AgdaSymbol{→}\AgdaSpace{}%
\AgdaBound{x}\AgdaSpace{}%
\AgdaOperator{\AgdaPostulate{≡≡}}\AgdaSpace{}%
\AgdaBound{y}\AgdaSpace{}%
\AgdaSymbol{→}\AgdaSpace{}%
\AgdaBound{y}\AgdaSpace{}%
\AgdaOperator{\AgdaPostulate{≡≡}}\AgdaSpace{}%
\AgdaBound{x}\<%
\\
\>[2]\AgdaFunction{symm}\AgdaSpace{}%
\AgdaBound{e}\AgdaSpace{}%
\AgdaSymbol{=}\AgdaSpace{}%
\AgdaPostulate{tpt}\AgdaSpace{}%
\AgdaSymbol{(λ}\AgdaSpace{}%
\AgdaBound{\AgdaUnderscore{}}\AgdaSpace{}%
\AgdaBound{y}\AgdaSpace{}%
\AgdaSymbol{→}\AgdaSpace{}%
\AgdaBound{y}\AgdaSpace{}%
\AgdaOperator{\AgdaPostulate{≡≡}}\AgdaSpace{}%
\AgdaSymbol{\AgdaUnderscore{})}\AgdaSpace{}%
\AgdaSymbol{(}\AgdaPostulate{eqt}\AgdaSpace{}%
\AgdaBound{e}\AgdaSymbol{)}\AgdaSpace{}%
\AgdaBound{e}\AgdaSpace{}%
\AgdaSymbol{(}\AgdaPostulate{rfl}\AgdaSpace{}%
\AgdaSymbol{\AgdaUnderscore{})}\<%
\\
\\[\AgdaEmptyExtraSkip]%
\>[2]\AgdaOperator{\AgdaFunction{proof\AgdaUnderscore{}}}\AgdaSpace{}%
\AgdaSymbol{:}\AgdaSpace{}%
\AgdaSymbol{∀\{}\AgdaBound{l}\AgdaSymbol{\}\{}\AgdaBound{A}\AgdaSpace{}%
\AgdaBound{B}\AgdaSpace{}%
\AgdaSymbol{:}\AgdaSpace{}%
\AgdaPrimitiveType{Set}\AgdaSpace{}%
\AgdaBound{l}\AgdaSymbol{\}\{}\AgdaBound{x}\AgdaSpace{}%
\AgdaSymbol{:}\AgdaSpace{}%
\AgdaBound{A}\AgdaSymbol{\}\{}\AgdaBound{y}\AgdaSpace{}%
\AgdaSymbol{:}\AgdaSpace{}%
\AgdaBound{B}\AgdaSymbol{\}}\AgdaSpace{}%
\AgdaSymbol{→}\AgdaSpace{}%
\AgdaBound{x}\AgdaSpace{}%
\AgdaOperator{\AgdaPostulate{≡≡}}\AgdaSpace{}%
\AgdaBound{y}\AgdaSpace{}%
\AgdaSymbol{→}\AgdaSpace{}%
\AgdaBound{x}\AgdaSpace{}%
\AgdaOperator{\AgdaPostulate{≡≡}}\AgdaSpace{}%
\AgdaBound{y}\<%
\\
\>[2]\AgdaOperator{\AgdaFunction{proof}}\AgdaSpace{}%
\AgdaBound{p}\AgdaSpace{}%
\AgdaSymbol{=}\AgdaSpace{}%
\AgdaBound{p}\<%
\\
\\[\AgdaEmptyExtraSkip]%
\>[2]\AgdaOperator{\AgdaFunction{\AgdaUnderscore{}≡≡[\AgdaUnderscore{}]\AgdaUnderscore{}}}%
\>[294I]\AgdaSymbol{:}\AgdaSpace{}%
\AgdaSymbol{∀\{}\AgdaBound{l}\AgdaSymbol{\}\{}\AgdaBound{A}\AgdaSpace{}%
\AgdaBound{B}\AgdaSpace{}%
\AgdaBound{C}\AgdaSpace{}%
\AgdaSymbol{:}\AgdaSpace{}%
\AgdaPrimitiveType{Set}\AgdaSpace{}%
\AgdaBound{l}\AgdaSymbol{\}(}\AgdaBound{x}\AgdaSpace{}%
\AgdaSymbol{:}\AgdaSpace{}%
\AgdaBound{A}\AgdaSymbol{)\{}\AgdaBound{y}\AgdaSpace{}%
\AgdaSymbol{:}\AgdaSpace{}%
\AgdaBound{B}\AgdaSymbol{\}\{}\AgdaBound{z}\AgdaSpace{}%
\AgdaSymbol{:}\AgdaSpace{}%
\AgdaBound{C}\AgdaSymbol{\}}\AgdaSpace{}%
\AgdaSymbol{→}\<%
\\
\>[.][@{}l@{}]\<[294I]%
\>[10]\AgdaBound{x}\AgdaSpace{}%
\AgdaOperator{\AgdaPostulate{≡≡}}\AgdaSpace{}%
\AgdaBound{y}%
\>[18]\AgdaSymbol{→}\AgdaSpace{}%
\AgdaBound{y}\AgdaSpace{}%
\AgdaOperator{\AgdaPostulate{≡≡}}\AgdaSpace{}%
\AgdaBound{z}%
\>[28]\AgdaSymbol{→}\AgdaSpace{}%
\AgdaBound{x}\AgdaSpace{}%
\AgdaOperator{\AgdaPostulate{≡≡}}\AgdaSpace{}%
\AgdaBound{z}\<%
\\
\>[2]\AgdaBound{x}\AgdaSpace{}%
\AgdaOperator{\AgdaFunction{≡≡[}}\AgdaSpace{}%
\AgdaBound{e}\AgdaSpace{}%
\AgdaOperator{\AgdaFunction{]}}\AgdaSpace{}%
\AgdaBound{f}\AgdaSpace{}%
\AgdaSymbol{=}\AgdaSpace{}%
\AgdaPostulate{tpt}\AgdaSpace{}%
\AgdaSymbol{(λ}\AgdaSpace{}%
\AgdaBound{\AgdaUnderscore{}}\AgdaSpace{}%
\AgdaBound{z}\AgdaSpace{}%
\AgdaSymbol{→}\AgdaSpace{}%
\AgdaBound{x}\AgdaSpace{}%
\AgdaOperator{\AgdaPostulate{≡≡}}\AgdaSpace{}%
\AgdaBound{z}\AgdaSymbol{)}\AgdaSpace{}%
\AgdaSymbol{(}\AgdaPostulate{eqt}\AgdaSpace{}%
\AgdaBound{f}\AgdaSymbol{)}\AgdaSpace{}%
\AgdaBound{f}\AgdaSpace{}%
\AgdaBound{e}\<%
\\
\\[\AgdaEmptyExtraSkip]%
\>[2]\AgdaOperator{\AgdaFunction{\AgdaUnderscore{}qed}}\AgdaSpace{}%
\AgdaSymbol{:}\AgdaSpace{}%
\AgdaSymbol{∀\{}\AgdaBound{l}\AgdaSymbol{\}\{}\AgdaBound{A}\AgdaSpace{}%
\AgdaSymbol{:}\AgdaSpace{}%
\AgdaPrimitiveType{Set}\AgdaSpace{}%
\AgdaBound{l}\AgdaSymbol{\}(}\AgdaBound{x}\AgdaSpace{}%
\AgdaSymbol{:}\AgdaSpace{}%
\AgdaBound{A}\AgdaSymbol{)}\AgdaSpace{}%
\AgdaSymbol{→}\AgdaSpace{}%
\AgdaBound{x}\AgdaSpace{}%
\AgdaOperator{\AgdaFunction{≡}}\AgdaSpace{}%
\AgdaBound{x}\<%
\\
\>[2]\AgdaBound{x}\AgdaSpace{}%
\AgdaOperator{\AgdaFunction{qed}}\AgdaSpace{}%
\AgdaSymbol{=}\AgdaSpace{}%
\AgdaPostulate{rfl}\AgdaSpace{}%
\AgdaBound{x}\<%
\\
\\[\AgdaEmptyExtraSkip]%
\>[2]\AgdaFunction{cong}\AgdaSpace{}%
\AgdaSymbol{:}%
\>[349I]\AgdaSymbol{∀\{}\AgdaBound{l}\AgdaSpace{}%
\AgdaBound{m}\AgdaSymbol{\}\{}\AgdaBound{A}\AgdaSpace{}%
\AgdaSymbol{:}\AgdaSpace{}%
\AgdaPrimitiveType{Set}\AgdaSpace{}%
\AgdaBound{l}\AgdaSymbol{\}\{}\AgdaBound{B}\AgdaSpace{}%
\AgdaSymbol{:}\AgdaSpace{}%
\AgdaBound{A}\AgdaSpace{}%
\AgdaSymbol{→}\AgdaSpace{}%
\AgdaPrimitiveType{Set}\AgdaSpace{}%
\AgdaBound{m}\AgdaSymbol{\}(}\AgdaBound{f}\AgdaSpace{}%
\AgdaSymbol{:}\AgdaSpace{}%
\AgdaSymbol{(}\AgdaBound{x}\AgdaSpace{}%
\AgdaSymbol{:}\AgdaSpace{}%
\AgdaBound{A}\AgdaSymbol{)}\AgdaSpace{}%
\AgdaSymbol{→}\AgdaSpace{}%
\AgdaBound{B}\AgdaSpace{}%
\AgdaBound{x}\AgdaSymbol{)\{}\AgdaBound{x}\AgdaSpace{}%
\AgdaBound{y}\AgdaSpace{}%
\AgdaSymbol{:}\AgdaSpace{}%
\AgdaBound{A}\AgdaSymbol{\}}\AgdaSpace{}%
\AgdaSymbol{→}\<%
\\
\>[.][@{}l@{}]\<[349I]%
\>[9]\AgdaBound{x}\AgdaSpace{}%
\AgdaOperator{\AgdaFunction{≡}}\AgdaSpace{}%
\AgdaBound{y}\AgdaSpace{}%
\AgdaSymbol{→}\AgdaSpace{}%
\AgdaBound{f}\AgdaSpace{}%
\AgdaBound{x}\AgdaSpace{}%
\AgdaOperator{\AgdaPostulate{≡≡}}\AgdaSpace{}%
\AgdaBound{f}\AgdaSpace{}%
\AgdaBound{y}\<%
\\
\>[2]\AgdaFunction{cong}\AgdaSpace{}%
\AgdaBound{f}\AgdaSpace{}%
\AgdaSymbol{\{}\AgdaBound{x}\AgdaSymbol{\}}\AgdaSpace{}%
\AgdaBound{e}\AgdaSpace{}%
\AgdaSymbol{=}\AgdaSpace{}%
\AgdaPostulate{tpt}\AgdaSpace{}%
\AgdaSymbol{(λ}\AgdaSpace{}%
\AgdaBound{\AgdaUnderscore{}}\AgdaSpace{}%
\AgdaBound{z}\AgdaSpace{}%
\AgdaSymbol{→}\AgdaSpace{}%
\AgdaBound{f}\AgdaSpace{}%
\AgdaBound{x}\AgdaSpace{}%
\AgdaOperator{\AgdaPostulate{≡≡}}\AgdaSpace{}%
\AgdaBound{f}\AgdaSpace{}%
\AgdaBound{z}\AgdaSymbol{)}\AgdaSpace{}%
\AgdaBound{e}\AgdaSpace{}%
\AgdaBound{e}\AgdaSpace{}%
\AgdaSymbol{(}\AgdaPostulate{rfl}\AgdaSpace{}%
\AgdaSymbol{(}\AgdaBound{f}\AgdaSpace{}%
\AgdaBound{x}\AgdaSymbol{))}\<%
\\
\\[\AgdaEmptyExtraSkip]%
\>[2]\AgdaFunction{cong₂}\AgdaSpace{}%
\AgdaSymbol{:}%
\>[398I]\AgdaSymbol{∀\{}\AgdaBound{l}\AgdaSpace{}%
\AgdaBound{m}\AgdaSpace{}%
\AgdaBound{n}\AgdaSymbol{\}\{}\AgdaBound{A}\AgdaSpace{}%
\AgdaSymbol{:}\AgdaSpace{}%
\AgdaPrimitiveType{Set}\AgdaSpace{}%
\AgdaBound{l}\AgdaSymbol{\}\{}\AgdaBound{B}\AgdaSpace{}%
\AgdaSymbol{:}\AgdaSpace{}%
\AgdaBound{A}\AgdaSpace{}%
\AgdaSymbol{→}\AgdaSpace{}%
\AgdaPrimitiveType{Set}\AgdaSpace{}%
\AgdaBound{m}\AgdaSymbol{\}\{}\AgdaBound{C}\AgdaSpace{}%
\AgdaSymbol{:}\AgdaSpace{}%
\AgdaSymbol{(}\AgdaBound{x}\AgdaSpace{}%
\AgdaSymbol{:}\AgdaSpace{}%
\AgdaBound{A}\AgdaSymbol{)}\AgdaSpace{}%
\AgdaSymbol{→}\AgdaSpace{}%
\AgdaBound{B}\AgdaSpace{}%
\AgdaBound{x}\AgdaSpace{}%
\AgdaSymbol{→}\AgdaSpace{}%
\AgdaPrimitiveType{Set}\AgdaSpace{}%
\AgdaBound{n}\AgdaSymbol{\}}\<%
\\
\>[.][@{}l@{}]\<[398I]%
\>[10]\AgdaSymbol{(}\AgdaBound{f}\AgdaSpace{}%
\AgdaSymbol{:}\AgdaSpace{}%
\AgdaSymbol{(}\AgdaBound{x}\AgdaSpace{}%
\AgdaSymbol{:}\AgdaSpace{}%
\AgdaBound{A}\AgdaSymbol{)(}\AgdaBound{y}\AgdaSpace{}%
\AgdaSymbol{:}\AgdaSpace{}%
\AgdaBound{B}\AgdaSpace{}%
\AgdaBound{x}\AgdaSymbol{)}\AgdaSpace{}%
\AgdaSymbol{→}\AgdaSpace{}%
\AgdaBound{C}\AgdaSpace{}%
\AgdaBound{x}\AgdaSpace{}%
\AgdaBound{y}\AgdaSymbol{)\{}\AgdaBound{x}\AgdaSpace{}%
\AgdaBound{x′}\AgdaSpace{}%
\AgdaSymbol{:}\AgdaSpace{}%
\AgdaBound{A}\AgdaSymbol{\}\{}\AgdaBound{y}\AgdaSpace{}%
\AgdaSymbol{:}\AgdaSpace{}%
\AgdaBound{B}\AgdaSpace{}%
\AgdaBound{x}\AgdaSymbol{\}\{}\AgdaBound{y′}\AgdaSpace{}%
\AgdaSymbol{:}\AgdaSpace{}%
\AgdaBound{B}\AgdaSpace{}%
\AgdaBound{x′}\AgdaSymbol{\}}\AgdaSpace{}%
\AgdaSymbol{→}\<%
\\
\>[10]\AgdaBound{x}\AgdaSpace{}%
\AgdaOperator{\AgdaFunction{≡}}\AgdaSpace{}%
\AgdaBound{x′}\AgdaSpace{}%
\AgdaSymbol{→}\AgdaSpace{}%
\AgdaBound{y}\AgdaSpace{}%
\AgdaOperator{\AgdaPostulate{≡≡}}\AgdaSpace{}%
\AgdaBound{y′}\AgdaSpace{}%
\AgdaSymbol{→}\AgdaSpace{}%
\AgdaBound{f}\AgdaSpace{}%
\AgdaBound{x}\AgdaSpace{}%
\AgdaBound{y}\AgdaSpace{}%
\AgdaOperator{\AgdaPostulate{≡≡}}\AgdaSpace{}%
\AgdaBound{f}\AgdaSpace{}%
\AgdaBound{x′}\AgdaSpace{}%
\AgdaBound{y′}\<%
\\
\>[2]\AgdaFunction{cong₂}\AgdaSpace{}%
\AgdaBound{f}\AgdaSpace{}%
\AgdaSymbol{\{}\AgdaBound{x}\AgdaSymbol{\}}\AgdaSpace{}%
\AgdaSymbol{\{\AgdaUnderscore{}\}}\AgdaSpace{}%
\AgdaSymbol{\{}\AgdaBound{y}\AgdaSymbol{\}}\AgdaSpace{}%
\AgdaBound{e}\AgdaSpace{}%
\AgdaBound{e′}\AgdaSpace{}%
\AgdaSymbol{=}\AgdaSpace{}%
\AgdaPostulate{tpt}\AgdaSpace{}%
\AgdaSymbol{(λ}\AgdaSpace{}%
\AgdaBound{x′}\AgdaSpace{}%
\AgdaBound{y′}\AgdaSpace{}%
\AgdaSymbol{→}\AgdaSpace{}%
\AgdaBound{f}\AgdaSpace{}%
\AgdaBound{x}\AgdaSpace{}%
\AgdaBound{y}\AgdaSpace{}%
\AgdaOperator{\AgdaPostulate{≡≡}}\AgdaSpace{}%
\AgdaBound{f}\AgdaSpace{}%
\AgdaBound{x′}\AgdaSpace{}%
\AgdaBound{y′}\AgdaSymbol{)}\AgdaSpace{}%
\AgdaBound{e}\AgdaSpace{}%
\AgdaBound{e′}\AgdaSpace{}%
\AgdaSymbol{(}\AgdaPostulate{rfl}\AgdaSpace{}%
\AgdaSymbol{(}\AgdaBound{f}\AgdaSpace{}%
\AgdaBound{x}\AgdaSpace{}%
\AgdaBound{y}\AgdaSymbol{))}\<%
\\
\\[\AgdaEmptyExtraSkip]%
\>[2]\AgdaFunction{cong₃}\AgdaSpace{}%
\AgdaSymbol{:}%
\>[480I]\AgdaSymbol{∀\{}\AgdaBound{k}\AgdaSpace{}%
\AgdaBound{l}\AgdaSpace{}%
\AgdaBound{m}\AgdaSpace{}%
\AgdaBound{n}\AgdaSymbol{\}\{}\AgdaBound{A}\AgdaSpace{}%
\AgdaSymbol{:}\AgdaSpace{}%
\AgdaPrimitiveType{Set}\AgdaSpace{}%
\AgdaBound{k}\AgdaSymbol{\}\{}\AgdaBound{B}\AgdaSpace{}%
\AgdaSymbol{:}\AgdaSpace{}%
\AgdaBound{A}\AgdaSpace{}%
\AgdaSymbol{→}\AgdaSpace{}%
\AgdaPrimitiveType{Set}\AgdaSpace{}%
\AgdaBound{l}\AgdaSymbol{\}\{}\AgdaBound{C}\AgdaSpace{}%
\AgdaSymbol{:}\AgdaSpace{}%
\AgdaSymbol{(}\AgdaBound{x}\AgdaSpace{}%
\AgdaSymbol{:}\AgdaSpace{}%
\AgdaBound{A}\AgdaSymbol{)}\AgdaSpace{}%
\AgdaSymbol{→}\AgdaSpace{}%
\AgdaBound{B}\AgdaSpace{}%
\AgdaBound{x}\AgdaSpace{}%
\AgdaSymbol{→}\AgdaSpace{}%
\AgdaPrimitiveType{Set}\AgdaSpace{}%
\AgdaBound{m}\AgdaSymbol{\}}\<%
\\
\>[.][@{}l@{}]\<[480I]%
\>[10]\AgdaSymbol{\{}\AgdaBound{D}\AgdaSpace{}%
\AgdaSymbol{:}\AgdaSpace{}%
\AgdaSymbol{(}\AgdaBound{x}\AgdaSpace{}%
\AgdaSymbol{:}\AgdaSpace{}%
\AgdaBound{A}\AgdaSymbol{)(}\AgdaBound{y}\AgdaSpace{}%
\AgdaSymbol{:}\AgdaSpace{}%
\AgdaBound{B}\AgdaSpace{}%
\AgdaBound{x}\AgdaSymbol{)}\AgdaSpace{}%
\AgdaSymbol{→}\AgdaSpace{}%
\AgdaBound{C}\AgdaSpace{}%
\AgdaBound{x}\AgdaSpace{}%
\AgdaBound{y}\AgdaSpace{}%
\AgdaSymbol{→}\AgdaSpace{}%
\AgdaPrimitiveType{Set}\AgdaSpace{}%
\AgdaBound{n}\AgdaSymbol{\}}\<%
\\
\>[10]\AgdaSymbol{(}\AgdaBound{f}\AgdaSpace{}%
\AgdaSymbol{:}\AgdaSpace{}%
\AgdaSymbol{(}\AgdaBound{x}\AgdaSpace{}%
\AgdaSymbol{:}\AgdaSpace{}%
\AgdaBound{A}\AgdaSymbol{)(}\AgdaBound{y}\AgdaSpace{}%
\AgdaSymbol{:}\AgdaSpace{}%
\AgdaBound{B}\AgdaSpace{}%
\AgdaBound{x}\AgdaSymbol{)(}\AgdaBound{z}\AgdaSpace{}%
\AgdaSymbol{:}\AgdaSpace{}%
\AgdaBound{C}\AgdaSpace{}%
\AgdaBound{x}\AgdaSpace{}%
\AgdaBound{y}\AgdaSymbol{)}\AgdaSpace{}%
\AgdaSymbol{→}\AgdaSpace{}%
\AgdaBound{D}\AgdaSpace{}%
\AgdaBound{x}\AgdaSpace{}%
\AgdaBound{y}\AgdaSpace{}%
\AgdaBound{z}\AgdaSymbol{)}\<%
\\
\>[10]\AgdaSymbol{\{}\AgdaBound{x}\AgdaSpace{}%
\AgdaBound{x′}\AgdaSpace{}%
\AgdaSymbol{:}\AgdaSpace{}%
\AgdaBound{A}\AgdaSymbol{\}\{}\AgdaBound{y}\AgdaSpace{}%
\AgdaSymbol{:}\AgdaSpace{}%
\AgdaBound{B}\AgdaSpace{}%
\AgdaBound{x}\AgdaSymbol{\}\{}\AgdaBound{y′}\AgdaSpace{}%
\AgdaSymbol{:}\AgdaSpace{}%
\AgdaBound{B}\AgdaSpace{}%
\AgdaBound{x′}\AgdaSymbol{\}\{}\AgdaBound{z}\AgdaSpace{}%
\AgdaSymbol{:}\AgdaSpace{}%
\AgdaBound{C}\AgdaSpace{}%
\AgdaBound{x}\AgdaSpace{}%
\AgdaBound{y}\AgdaSymbol{\}\{}\AgdaBound{z′}\AgdaSpace{}%
\AgdaSymbol{:}\AgdaSpace{}%
\AgdaBound{C}\AgdaSpace{}%
\AgdaBound{x′}\AgdaSpace{}%
\AgdaBound{y′}\AgdaSymbol{\}}\AgdaSpace{}%
\AgdaSymbol{→}\<%
\\
\>[10]\AgdaBound{x}\AgdaSpace{}%
\AgdaOperator{\AgdaFunction{≡}}\AgdaSpace{}%
\AgdaBound{x′}\AgdaSpace{}%
\AgdaSymbol{→}\AgdaSpace{}%
\AgdaBound{y}\AgdaSpace{}%
\AgdaOperator{\AgdaPostulate{≡≡}}\AgdaSpace{}%
\AgdaBound{y′}\AgdaSpace{}%
\AgdaSymbol{→}\AgdaSpace{}%
\AgdaBound{z}\AgdaSpace{}%
\AgdaOperator{\AgdaPostulate{≡≡}}\AgdaSpace{}%
\AgdaBound{z′}\AgdaSpace{}%
\AgdaSymbol{→}\AgdaSpace{}%
\AgdaBound{f}\AgdaSpace{}%
\AgdaBound{x}\AgdaSpace{}%
\AgdaBound{y}\AgdaSpace{}%
\AgdaBound{z}\AgdaSpace{}%
\AgdaOperator{\AgdaPostulate{≡≡}}\AgdaSpace{}%
\AgdaBound{f}\AgdaSpace{}%
\AgdaBound{x′}\AgdaSpace{}%
\AgdaBound{y′}\AgdaSpace{}%
\AgdaBound{z′}\<%
\\
\>[2]\AgdaFunction{cong₃}\AgdaSpace{}%
\AgdaBound{f}\AgdaSpace{}%
\AgdaSymbol{\{}\AgdaBound{x}\AgdaSymbol{\}}\AgdaSpace{}%
\AgdaSymbol{\{\AgdaUnderscore{}\}}\AgdaSpace{}%
\AgdaSymbol{\{}\AgdaBound{y}\AgdaSymbol{\}}\AgdaSpace{}%
\AgdaSymbol{\{\AgdaUnderscore{}\}}\AgdaSpace{}%
\AgdaSymbol{\{}\AgdaBound{z}\AgdaSymbol{\}}\AgdaSpace{}%
\AgdaBound{e}\AgdaSpace{}%
\AgdaBound{e′}\AgdaSpace{}%
\AgdaSymbol{=}\<%
\\
\>[2][@{}l@{\AgdaIndent{0}}]%
\>[4]\AgdaPostulate{tpt}\AgdaSpace{}%
\AgdaSymbol{(λ}\AgdaSpace{}%
\AgdaBound{x′}\AgdaSpace{}%
\AgdaBound{y′}\AgdaSpace{}%
\AgdaSymbol{→}\AgdaSpace{}%
\AgdaSymbol{∀\{}\AgdaBound{z′}\AgdaSymbol{\}}\AgdaSpace{}%
\AgdaSymbol{→}\AgdaSpace{}%
\AgdaBound{z}\AgdaSpace{}%
\AgdaOperator{\AgdaPostulate{≡≡}}\AgdaSpace{}%
\AgdaBound{z′}\AgdaSpace{}%
\AgdaSymbol{→}\AgdaSpace{}%
\AgdaBound{f}\AgdaSpace{}%
\AgdaBound{x}\AgdaSpace{}%
\AgdaBound{y}\AgdaSpace{}%
\AgdaBound{z}\AgdaSpace{}%
\AgdaOperator{\AgdaPostulate{≡≡}}\AgdaSpace{}%
\AgdaBound{f}\AgdaSpace{}%
\AgdaBound{x′}\AgdaSpace{}%
\AgdaBound{y′}\AgdaSpace{}%
\AgdaBound{z′}\AgdaSymbol{)}\AgdaSpace{}%
\AgdaBound{e}\AgdaSpace{}%
\AgdaBound{e′}\AgdaSpace{}%
\AgdaSymbol{(}\AgdaFunction{cong₂}\AgdaSpace{}%
\AgdaSymbol{(}\AgdaBound{f}\AgdaSpace{}%
\AgdaBound{x}\AgdaSymbol{)}\AgdaSpace{}%
\AgdaSymbol{(}\AgdaPostulate{rfl}\AgdaSpace{}%
\AgdaBound{y}\AgdaSymbol{))}\<%
\end{code}
\caption{Equational reasoning for heterogeneous equality}
\label{fig:equrheq}
\end{figure}

We begin with some simple lemmas establishing the basics of equational
logic for $\heq$, namely chain-reasoning using reflexivity (already an
axiom), symmetry, transitivity and congruence properties. These are
given in Fig.~\ref{fig:equrheq}.

The axioms in Fig.~\ref{fig:axihtpe} are notably lacking a
``regularity'' property for $\afun{tpt}$, that is, a proof of type
$\afun{tpt}\,(\afun{rfl}\,\aarg{x})\,(\afun{rfl}\,\aarg{y})\,\aarg{z}
\eq \aarg{z}$. But such a thing is needed if we are to derive the
expected elimination and (typal) computation rules for $\_{\heq}\_$
and $\_{\eq}\_$. To get those, one can define a ``corrected'' form of
transport that has this regularity property, using a simplified
version of a trick due to Peter Lumsdaine~[unpublished]. In fact, it
is enough to produce a function coercing proofs of equality of types
$\aarg{e}:\aarg{A}\eq\aarg{B}$ into functions
$\afun{coe}\,\aarg{e} : \aarg{A}\fun\aarg{B}$ and which satisfies the
heterogeneous regularity property that
$\afun{coe}\,\aarg{e}\,\aarg{x} \heq \aarg{x}$ (so that, given how we
define $\eq$ in terms of $\heq$, the usual form of regularity,
$\afun{coe}\,(\afun{rfl}\,\aarg{A})\,\aarg{x} \eq \aarg{x}$, is just
the special case of this when $\aarg{e}$ is $\afun{rfl}\,\aarg{A}$).

\begin{lemma}
  \label{lem:verlt}
  The axioms in Fig.~\ref{fig:axihtpe} imply the existence of a coercion
  function
\AgdaTarget{coe}
{\normalfont\begin{code}%
\>[2]\AgdaFunction{coe}\AgdaSpace{}%
\AgdaSymbol{:}\AgdaSpace{}%
\AgdaSymbol{∀\{}\AgdaBound{l}\AgdaSymbol{\}\{}\AgdaBound{A}\AgdaSpace{}%
\AgdaBound{B}\AgdaSpace{}%
\AgdaSymbol{:}\AgdaSpace{}%
\AgdaPrimitiveType{Set}\AgdaSpace{}%
\AgdaBound{l}\AgdaSymbol{\}}\AgdaSpace{}%
\AgdaSymbol{→}\AgdaSpace{}%
\AgdaBound{A}\AgdaSpace{}%
\AgdaOperator{\AgdaFunction{≡}}\AgdaSpace{}%
\AgdaBound{B}\AgdaSpace{}%
\AgdaSymbol{→}\AgdaSpace{}%
\AgdaBound{A}\AgdaSpace{}%
\AgdaSymbol{→}\AgdaSpace{}%
\AgdaBound{B}\<%
\end{code}}
  satisfying a heterogeneous regularity property:
\AgdaTarget{coeIsRegular}
{\normalfont\begin{code}%
\>[2]\AgdaFunction{coeIsRegular}\AgdaSpace{}%
\AgdaSymbol{:}\AgdaSpace{}%
\AgdaSymbol{∀\{}\AgdaBound{l}\AgdaSymbol{\}\{}\AgdaBound{A}\AgdaSpace{}%
\AgdaBound{B}\AgdaSpace{}%
\AgdaSymbol{:}\AgdaSpace{}%
\AgdaPrimitiveType{Set}\AgdaSpace{}%
\AgdaBound{l}\AgdaSymbol{\}(}\AgdaBound{e}\AgdaSpace{}%
\AgdaSymbol{:}\AgdaSpace{}%
\AgdaBound{A}\AgdaSpace{}%
\AgdaOperator{\AgdaFunction{≡}}\AgdaSpace{}%
\AgdaBound{B}\AgdaSymbol{)(}\AgdaBound{x}\AgdaSpace{}%
\AgdaSymbol{:}\AgdaSpace{}%
\AgdaBound{A}\AgdaSymbol{)}\AgdaSpace{}%
\AgdaSymbol{→}\AgdaSpace{}%
\AgdaFunction{coe}\AgdaSpace{}%
\AgdaBound{e}\AgdaSpace{}%
\AgdaBound{x}\AgdaSpace{}%
\AgdaOperator{\AgdaPostulate{≡≡}}\AgdaSpace{}%
\AgdaBound{x}\<%
\end{code}}
\end{lemma}
\begin{proof}
First we define the type of
functions that are injective with respect to $\eq$ and note that the
identity function is one such:
\AgdaTarget{Inj}
\AgdaTarget{id}
\AgdaTarget{idInj}
\begin{code}%
\>[2]\AgdaFunction{Inj}%
\>[12]\AgdaSymbol{:}\AgdaSpace{}%
\AgdaSymbol{∀\{}\AgdaBound{l}\AgdaSymbol{\}(}\AgdaBound{A}\AgdaSpace{}%
\AgdaBound{B}\AgdaSpace{}%
\AgdaSymbol{:}\AgdaSpace{}%
\AgdaPrimitiveType{Set}\AgdaSpace{}%
\AgdaBound{l}\AgdaSymbol{)}\AgdaSpace{}%
\AgdaSymbol{→}\AgdaSpace{}%
\AgdaPrimitiveType{Set}\AgdaSpace{}%
\AgdaBound{l}\<%
\\
\>[2]\AgdaFunction{Inj}\AgdaSpace{}%
\AgdaBound{A}\AgdaSpace{}%
\AgdaBound{B}%
\>[12]\AgdaSymbol{=}\AgdaSpace{}%
\AgdaPostulate{∑}\AgdaSpace{}%
\AgdaBound{f}\AgdaSpace{}%
\AgdaPostulate{∶}\AgdaSpace{}%
\AgdaSymbol{(}\AgdaBound{A}\AgdaSpace{}%
\AgdaSymbol{→}\AgdaSpace{}%
\AgdaBound{B}\AgdaSymbol{)}\AgdaSpace{}%
\AgdaPostulate{,}\AgdaSpace{}%
\AgdaSymbol{∀\{}\AgdaBound{x}\AgdaSpace{}%
\AgdaBound{y}\AgdaSymbol{\}}\AgdaSpace{}%
\AgdaSymbol{→}\AgdaSpace{}%
\AgdaBound{f}\AgdaSpace{}%
\AgdaBound{x}\AgdaSpace{}%
\AgdaOperator{\AgdaFunction{≡}}\AgdaSpace{}%
\AgdaBound{f}\AgdaSpace{}%
\AgdaBound{y}\AgdaSpace{}%
\AgdaSymbol{→}\AgdaSpace{}%
\AgdaBound{x}\AgdaSpace{}%
\AgdaOperator{\AgdaFunction{≡}}\AgdaSpace{}%
\AgdaBound{y}\<%
\\
\\[\AgdaEmptyExtraSkip]%
\>[2]\AgdaFunction{id}%
\>[12]\AgdaSymbol{:}\AgdaSpace{}%
\AgdaSymbol{∀\{}\AgdaBound{l}\AgdaSymbol{\}\{}\AgdaBound{A}\AgdaSpace{}%
\AgdaSymbol{:}\AgdaSpace{}%
\AgdaPrimitiveType{Set}\AgdaSpace{}%
\AgdaBound{l}\AgdaSymbol{\}}\AgdaSpace{}%
\AgdaSymbol{→}\AgdaSpace{}%
\AgdaBound{A}\AgdaSpace{}%
\AgdaSymbol{→}\AgdaSpace{}%
\AgdaBound{A}\<%
\\
\>[2]\AgdaFunction{id}\AgdaSpace{}%
\AgdaBound{x}%
\>[12]\AgdaSymbol{=}\AgdaSpace{}%
\AgdaBound{x}\<%
\\
\>[0]\<%
\\
\>[2]\AgdaFunction{idInj}%
\>[12]\AgdaSymbol{:}\AgdaSpace{}%
\AgdaSymbol{∀\{}\AgdaBound{l}\AgdaSymbol{\}(}\AgdaBound{A}\AgdaSpace{}%
\AgdaSymbol{:}\AgdaSpace{}%
\AgdaPrimitiveType{Set}\AgdaSpace{}%
\AgdaBound{l}\AgdaSymbol{)}\AgdaSpace{}%
\AgdaSymbol{→}\AgdaSpace{}%
\AgdaFunction{Inj}\AgdaSpace{}%
\AgdaBound{A}\AgdaSpace{}%
\AgdaBound{A}\<%
\\
\>[2]\AgdaFunction{idInj}\AgdaSpace{}%
\AgdaSymbol{\AgdaUnderscore{}}%
\>[12]\AgdaSymbol{=}\AgdaSpace{}%
\AgdaSymbol{(}\AgdaFunction{id}\AgdaSpace{}%
\AgdaOperator{\AgdaPostulate{,}}\AgdaSpace{}%
\AgdaFunction{id}\AgdaSymbol{)}\<%
\end{code}
Next we use $\afun{tpt}$ to define a function coercing equalities into
injective functions:
\AgdaTarget{icoe}
\begin{code}%
\>[2]\AgdaFunction{icoe}\AgdaSpace{}%
\AgdaSymbol{:}\AgdaSpace{}%
\AgdaSymbol{∀\{}\AgdaBound{l}\AgdaSymbol{\}\{}\AgdaBound{A}\AgdaSpace{}%
\AgdaBound{B}\AgdaSpace{}%
\AgdaSymbol{:}\AgdaSpace{}%
\AgdaPrimitiveType{Set}\AgdaSpace{}%
\AgdaBound{l}\AgdaSymbol{\}}\AgdaSpace{}%
\AgdaSymbol{→}\AgdaSpace{}%
\AgdaBound{A}\AgdaSpace{}%
\AgdaOperator{\AgdaFunction{≡}}\AgdaSpace{}%
\AgdaBound{B}\AgdaSpace{}%
\AgdaSymbol{→}\AgdaSpace{}%
\AgdaFunction{Inj}\AgdaSpace{}%
\AgdaBound{A}\AgdaSpace{}%
\AgdaBound{B}\<%
\\
\>[2]\AgdaFunction{icoe}\AgdaSpace{}%
\AgdaSymbol{\{}\AgdaBound{l}\AgdaSymbol{\}}\AgdaSpace{}%
\AgdaSymbol{\{}\AgdaBound{A}\AgdaSymbol{\}}\AgdaSpace{}%
\AgdaBound{e}\AgdaSpace{}%
\AgdaSymbol{=}\AgdaSpace{}%
\AgdaPostulate{tpt}\AgdaSpace{}%
\AgdaSymbol{(λ}\AgdaSpace{}%
\AgdaBound{\AgdaUnderscore{}}\AgdaSpace{}%
\AgdaBound{C}\AgdaSpace{}%
\AgdaSymbol{→}\AgdaSpace{}%
\AgdaFunction{Inj}\AgdaSpace{}%
\AgdaBound{A}\AgdaSpace{}%
\AgdaBound{C}\AgdaSymbol{)}\AgdaSpace{}%
\AgdaSymbol{(}\AgdaPostulate{rfl}\AgdaSpace{}%
\AgdaSymbol{(}\AgdaPrimitiveType{Set}\AgdaSpace{}%
\AgdaBound{l}\AgdaSymbol{))}\AgdaSpace{}%
\AgdaBound{e}\AgdaSpace{}%
\AgdaSymbol{(}\AgdaFunction{idInj}\AgdaSpace{}%
\AgdaBound{A}\AgdaSymbol{)}\<%
\end{code}
The injectiveness of $\afun{icoe}\,\aarg{e}$ is used as follows. Applying the operation
$\afun{tpt}$ to the type family
\AgdaTarget{fsticoe}
\begin{code}%
\>[2]\AgdaFunction{fsticoe}\AgdaSpace{}%
\AgdaSymbol{:}\AgdaSpace{}%
\AgdaSymbol{∀\{}\AgdaBound{l}\AgdaSymbol{\}\{}\AgdaBound{A}\AgdaSpace{}%
\AgdaSymbol{:}\AgdaSpace{}%
\AgdaPrimitiveType{Set}\AgdaSpace{}%
\AgdaBound{l}\AgdaSymbol{\}(}\AgdaBound{x}\AgdaSpace{}%
\AgdaSymbol{:}\AgdaSpace{}%
\AgdaBound{A}\AgdaSymbol{)(}\AgdaBound{B}\AgdaSpace{}%
\AgdaSymbol{:}\AgdaSpace{}%
\AgdaPrimitiveType{Set}\AgdaSpace{}%
\AgdaBound{l}\AgdaSymbol{)(}\AgdaBound{e}\AgdaSpace{}%
\AgdaSymbol{:}\AgdaSpace{}%
\AgdaBound{A}\AgdaSpace{}%
\AgdaOperator{\AgdaFunction{≡}}\AgdaSpace{}%
\AgdaBound{B}\AgdaSymbol{)}\AgdaSpace{}%
\AgdaSymbol{→}\AgdaSpace{}%
\AgdaPrimitiveType{Set}\AgdaSpace{}%
\AgdaBound{l}\<%
\\
\>[2]\AgdaFunction{fsticoe}\AgdaSpace{}%
\AgdaBound{x}\AgdaSpace{}%
\AgdaBound{B}\AgdaSpace{}%
\AgdaBound{e}\AgdaSpace{}%
\AgdaSymbol{=}\AgdaSpace{}%
\AgdaPostulate{∑}\AgdaSpace{}%
\AgdaBound{y}\AgdaSpace{}%
\AgdaPostulate{∶}\AgdaSpace{}%
\AgdaBound{B}%
\>[27]\AgdaPostulate{,}\AgdaSpace{}%
\AgdaSymbol{(}\AgdaPostulate{fst}\AgdaSpace{}%
\AgdaSymbol{(}\AgdaFunction{icoe}\AgdaSpace{}%
\AgdaSymbol{(}\AgdaPostulate{rfl}\AgdaSpace{}%
\AgdaBound{B}\AgdaSymbol{))}\AgdaSpace{}%
\AgdaBound{y}\AgdaSpace{}%
\AgdaOperator{\AgdaFunction{≡}}\AgdaSpace{}%
\AgdaPostulate{fst}\AgdaSpace{}%
\AgdaSymbol{(}\AgdaFunction{icoe}\AgdaSpace{}%
\AgdaBound{e}\AgdaSymbol{)}\AgdaSpace{}%
\AgdaBound{x}\AgdaSymbol{)}\<%
\end{code}
we can transport the element
$(\aarg{e}\,,\,\afun{rfl}\, (\afun{fst}\,
(\afun{icoe}\,(\afun{rfl}\,\aarg{A}))\, \aarg{x})$ of type
$\afun{fsticoe}\,\aarg{x}\,\aarg{A}\,(\afun{rfl}\,\aarg{A})$ along
$\aarg{e}:\aarg{A}\eq\aarg{B}$ and
$\afun{ctr}\,\aarg{e}:\afun{rfl}\,\aarg{A}\heq \aarg{e}$ to give an
element of type $\afun{fsticoe}\,\aarg{x}\,\aarg{B}\,\aarg{e}$.  The
first projection of this element gives the value of the desired
coercion along $\aarg{e}$ at $\aarg{x}\,$:
\begin{code}%
\>[2]\AgdaFunction{coe}\AgdaSpace{}%
\AgdaBound{e}\AgdaSpace{}%
\AgdaBound{x}\AgdaSpace{}%
\AgdaSymbol{=}\AgdaSpace{}%
\AgdaPostulate{fst}\AgdaSpace{}%
\AgdaSymbol{(}\AgdaPostulate{tpt}\AgdaSpace{}%
\AgdaSymbol{(}\AgdaFunction{fsticoe}\AgdaSpace{}%
\AgdaBound{x}\AgdaSymbol{)}\AgdaSpace{}%
\AgdaBound{e}\AgdaSpace{}%
\AgdaSymbol{(}\AgdaPostulate{ctr}\AgdaSpace{}%
\AgdaBound{e}\AgdaSymbol{)}\AgdaSpace{}%
\AgdaSymbol{(}\AgdaBound{x}\AgdaSpace{}%
\AgdaOperator{\AgdaPostulate{,}}\AgdaSpace{}%
\AgdaPostulate{rfl}\AgdaSpace{}%
\AgdaSymbol{\AgdaUnderscore{}))}\<%
\end{code}
and its second projection can
be used along with the injectiveness property of $\afun{icoe}$ to get
the regularity property of this coercion:
\begin{code}%
\>[2]\AgdaFunction{coeIsRegular}\AgdaSpace{}%
\AgdaSymbol{\{\AgdaUnderscore{}\}}\AgdaSpace{}%
\AgdaSymbol{\{}\AgdaBound{A}\AgdaSymbol{\}}\AgdaSpace{}%
\AgdaBound{e}\AgdaSpace{}%
\AgdaBound{x}\AgdaSpace{}%
\AgdaSymbol{=}\AgdaSpace{}%
\AgdaPostulate{tpt}\AgdaSpace{}%
\AgdaSymbol{(λ}\AgdaSpace{}%
\AgdaBound{\AgdaUnderscore{}}\AgdaSpace{}%
\AgdaBound{e′}\AgdaSpace{}%
\AgdaSymbol{→}\AgdaSpace{}%
\AgdaFunction{coe}\AgdaSpace{}%
\AgdaBound{e′}\AgdaSpace{}%
\AgdaBound{x}\AgdaSpace{}%
\AgdaOperator{\AgdaPostulate{≡≡}}\AgdaSpace{}%
\AgdaBound{x}\AgdaSymbol{)}\AgdaSpace{}%
\AgdaBound{e}\AgdaSpace{}%
\AgdaSymbol{(}\AgdaPostulate{ctr}%
\>[66]\AgdaBound{e}\AgdaSymbol{)}\AgdaSpace{}%
\AgdaFunction{coerfl}\<%
\\
\>[2][@{}l@{\AgdaIndent{0}}]%
\>[4]\AgdaKeyword{where}\<%
\\
\>[4]\AgdaFunction{coerfl}\AgdaSpace{}%
\AgdaSymbol{:}\AgdaSpace{}%
\AgdaFunction{coe}\AgdaSpace{}%
\AgdaSymbol{(}\AgdaPostulate{rfl}\AgdaSpace{}%
\AgdaBound{A}\AgdaSymbol{)}\AgdaSpace{}%
\AgdaBound{x}\AgdaSpace{}%
\AgdaOperator{\AgdaFunction{≡}}\AgdaSpace{}%
\AgdaBound{x}\<%
\\
\>[4]\AgdaFunction{coerfl}\AgdaSpace{}%
\AgdaSymbol{=}%
\>[795I]\AgdaPostulate{snd}\AgdaSpace{}%
\AgdaSymbol{(}\AgdaFunction{icoe}\AgdaSpace{}%
\AgdaSymbol{(}\AgdaPostulate{rfl}\AgdaSpace{}%
\AgdaBound{A}\AgdaSymbol{))}\AgdaSpace{}%
\AgdaSymbol{(}\AgdaPostulate{snd}\AgdaSpace{}%
\AgdaSymbol{(}\<%
\\
\>[795I][@{}l@{\AgdaIndent{0}}]%
\>[15]\AgdaPostulate{tpt}\AgdaSpace{}%
\AgdaSymbol{(}\AgdaFunction{fsticoe}\AgdaSpace{}%
\AgdaBound{x}\AgdaSymbol{)}\AgdaSpace{}%
\AgdaSymbol{(}\AgdaPostulate{rfl}\AgdaSpace{}%
\AgdaBound{A}\AgdaSymbol{)}\AgdaSpace{}%
\AgdaSymbol{(}\AgdaPostulate{ctr}\AgdaSpace{}%
\AgdaSymbol{(}\AgdaPostulate{rfl}\AgdaSpace{}%
\AgdaBound{A}\AgdaSymbol{))}\AgdaSpace{}%
\AgdaSymbol{(}\AgdaBound{x}\AgdaSpace{}%
\AgdaOperator{\AgdaPostulate{,}}\AgdaSpace{}%
\AgdaPostulate{rfl}\AgdaSpace{}%
\AgdaSymbol{\AgdaUnderscore{})))}\<%
\end{code}
\end{proof}

An immediate corollary is that the axioms imply the \emph{uniqueness
  of identity proofs} (UIP) and hence Streicher's
Axiom~K~\citep{SteicherT:invitt}. (We will see in Sect.~\ref{sec:cona}
that in fact it is only the $\afun{tpt}$ function that contains an
implicit use of Axiom~K.)

\begin{theorem}[UIP and Axiom K]
  \label{thm:uipak}
  The axioms in Fig.~\ref{fig:axihtpe} imply that $\eq$ satisfies
  {\normalfont
\AgdaTarget{uip}
\AgdaTarget{axiomK}
\AgdaTarget{axiomKComp}
\begin{code}%
\>[2]\AgdaFunction{uip}%
\>[14]\AgdaSymbol{:}\AgdaSpace{}%
\AgdaSymbol{∀\{}\AgdaBound{l}\AgdaSymbol{\}\{}\AgdaBound{A}\AgdaSpace{}%
\AgdaSymbol{:}\AgdaSpace{}%
\AgdaPrimitiveType{Set}\AgdaSpace{}%
\AgdaBound{l}\AgdaSymbol{\}\{}\AgdaBound{x}\AgdaSpace{}%
\AgdaBound{y}\AgdaSpace{}%
\AgdaSymbol{:}\AgdaSpace{}%
\AgdaBound{A}\AgdaSymbol{\}(}\AgdaBound{e}\AgdaSpace{}%
\AgdaBound{e′}\AgdaSpace{}%
\AgdaSymbol{:}\AgdaSpace{}%
\AgdaBound{x}\AgdaSpace{}%
\AgdaOperator{\AgdaFunction{≡}}\AgdaSpace{}%
\AgdaBound{y}\AgdaSymbol{)}\AgdaSpace{}%
\AgdaSymbol{→}\AgdaSpace{}%
\AgdaBound{e}\AgdaSpace{}%
\AgdaOperator{\AgdaFunction{≡}}\AgdaSpace{}%
\AgdaBound{e′}\<%
\\
\>[0]\<%
\\
\>[2]\AgdaFunction{axiomK}%
\>[14]\AgdaSymbol{:}%
\>[828I]\AgdaSymbol{∀\{}\AgdaBound{l}\AgdaSpace{}%
\AgdaBound{m}\AgdaSymbol{\}\{}\AgdaBound{A}\AgdaSpace{}%
\AgdaSymbol{:}\AgdaSpace{}%
\AgdaPrimitiveType{Set}\AgdaSpace{}%
\AgdaBound{l}\AgdaSymbol{\}\{}\AgdaBound{x}\AgdaSpace{}%
\AgdaSymbol{:}\AgdaSpace{}%
\AgdaBound{A}\AgdaSymbol{\}(}\AgdaBound{P}\AgdaSpace{}%
\AgdaSymbol{:}\AgdaSpace{}%
\AgdaBound{x}\AgdaSpace{}%
\AgdaOperator{\AgdaFunction{≡}}\AgdaSpace{}%
\AgdaBound{x}\AgdaSpace{}%
\AgdaSymbol{→}\AgdaSpace{}%
\AgdaPrimitiveType{Set}\AgdaSpace{}%
\AgdaBound{m}\AgdaSymbol{)(}\AgdaBound{p}\AgdaSpace{}%
\AgdaSymbol{:}\AgdaSpace{}%
\AgdaBound{P}\AgdaSpace{}%
\AgdaSymbol{(}\AgdaPostulate{rfl}\AgdaSpace{}%
\AgdaBound{x}\AgdaSymbol{))}\AgdaSpace{}%
\AgdaSymbol{→}\<%
\\
\>[.][@{}l@{}]\<[828I]%
\>[16]\AgdaSymbol{∀}\AgdaSpace{}%
\AgdaBound{e}\AgdaSpace{}%
\AgdaSymbol{→}\AgdaSpace{}%
\AgdaBound{P}\AgdaSpace{}%
\AgdaBound{e}\<%
\\
\\[\AgdaEmptyExtraSkip]%
\>[2]\AgdaFunction{axiomKComp}%
\>[14]\AgdaSymbol{:}%
\>[851I]\AgdaSymbol{∀\{}\AgdaBound{l}\AgdaSpace{}%
\AgdaBound{m}\AgdaSymbol{\}\{}\AgdaBound{A}\AgdaSpace{}%
\AgdaSymbol{:}\AgdaSpace{}%
\AgdaPrimitiveType{Set}\AgdaSpace{}%
\AgdaBound{l}\AgdaSymbol{\}\{}\AgdaBound{x}\AgdaSpace{}%
\AgdaSymbol{:}\AgdaSpace{}%
\AgdaBound{A}\AgdaSymbol{\}(}\AgdaBound{P}\AgdaSpace{}%
\AgdaSymbol{:}\AgdaSpace{}%
\AgdaBound{x}\AgdaSpace{}%
\AgdaOperator{\AgdaFunction{≡}}\AgdaSpace{}%
\AgdaBound{x}\AgdaSpace{}%
\AgdaSymbol{→}\AgdaSpace{}%
\AgdaPrimitiveType{Set}\AgdaSpace{}%
\AgdaBound{m}\AgdaSymbol{)(}\AgdaBound{p}\AgdaSpace{}%
\AgdaSymbol{:}\AgdaSpace{}%
\AgdaBound{P}\AgdaSpace{}%
\AgdaSymbol{(}\AgdaPostulate{rfl}\AgdaSpace{}%
\AgdaBound{x}\AgdaSymbol{))}\AgdaSpace{}%
\AgdaSymbol{→}\<%
\\
\>[.][@{}l@{}]\<[851I]%
\>[16]\AgdaFunction{axiomK}\AgdaSpace{}%
\AgdaBound{P}\AgdaSpace{}%
\AgdaBound{p}\AgdaSpace{}%
\AgdaSymbol{(}\AgdaPostulate{rfl}\AgdaSpace{}%
\AgdaBound{x}\AgdaSymbol{)}\AgdaSpace{}%
\AgdaOperator{\AgdaFunction{≡}}\AgdaSpace{}%
\AgdaBound{p}\<%
\end{code}}
\end{theorem}
\begin{proof}
Using the functions from Fig~\ref{fig:equrheq} and Lemma~\ref{lem:verlt} we have:
\begin{code}%
\>[2]\AgdaFunction{uip}\AgdaSpace{}%
\AgdaBound{e}\AgdaSpace{}%
\AgdaBound{e′}%
\>[18]\AgdaSymbol{=}\AgdaSpace{}%
\AgdaPostulate{tpt}\AgdaSpace{}%
\AgdaSymbol{(λ}\AgdaSpace{}%
\AgdaBound{\AgdaUnderscore{}}\AgdaSpace{}%
\AgdaBound{e′′}\AgdaSpace{}%
\AgdaSymbol{→}\AgdaSpace{}%
\AgdaBound{e′′}\AgdaSpace{}%
\AgdaOperator{\AgdaPostulate{≡≡}}\AgdaSpace{}%
\AgdaBound{e′}\AgdaSymbol{)}\AgdaSpace{}%
\AgdaBound{e}\AgdaSpace{}%
\AgdaSymbol{(}\AgdaPostulate{ctr}\AgdaSpace{}%
\AgdaBound{e}\AgdaSymbol{)}\AgdaSpace{}%
\AgdaSymbol{(}\AgdaPostulate{ctr}\AgdaSpace{}%
\AgdaBound{e′}\AgdaSymbol{)}\<%
\\
\\[\AgdaEmptyExtraSkip]%
\>[2]\AgdaFunction{axiomK}\AgdaSpace{}%
\AgdaBound{P}\AgdaSpace{}%
\AgdaBound{p}\AgdaSpace{}%
\AgdaBound{e}%
\>[18]\AgdaSymbol{=}\AgdaSpace{}%
\AgdaFunction{coe}\AgdaSpace{}%
\AgdaSymbol{(}\AgdaFunction{cong₂}\AgdaSpace{}%
\AgdaSymbol{(λ}\AgdaSpace{}%
\AgdaBound{\AgdaUnderscore{}}\AgdaSpace{}%
\AgdaSymbol{→}\AgdaSpace{}%
\AgdaBound{P}\AgdaSymbol{)}\AgdaSpace{}%
\AgdaSymbol{(}\AgdaPostulate{rfl}\AgdaSpace{}%
\AgdaBound{p}\AgdaSymbol{)}\AgdaSpace{}%
\AgdaSymbol{(}\AgdaPostulate{ctr}\AgdaSpace{}%
\AgdaBound{e}\AgdaSymbol{))}\AgdaSpace{}%
\AgdaBound{p}\<%
\\
\\[\AgdaEmptyExtraSkip]%
\>[2]\AgdaFunction{axiomKComp}\AgdaSpace{}%
\AgdaBound{P}\AgdaSpace{}%
\AgdaBound{p}%
\>[18]\AgdaSymbol{=}%
\>[21]\AgdaOperator{\AgdaFunction{proof}}\<%
\\
\>[21][@{}l@{\AgdaIndent{0}}]%
\>[23]\AgdaFunction{coe}\AgdaSpace{}%
\AgdaSymbol{(}\AgdaFunction{cong₂}\AgdaSpace{}%
\AgdaSymbol{(λ}\AgdaSpace{}%
\AgdaBound{\AgdaUnderscore{}}\AgdaSpace{}%
\AgdaSymbol{→}\AgdaSpace{}%
\AgdaBound{P}\AgdaSymbol{)}\AgdaSpace{}%
\AgdaSymbol{(}\AgdaPostulate{rfl}\AgdaSpace{}%
\AgdaBound{p}\AgdaSymbol{)}\AgdaSpace{}%
\AgdaSymbol{(}\AgdaPostulate{ctr}\AgdaSpace{}%
\AgdaSymbol{(}\AgdaPostulate{rfl}\AgdaSpace{}%
\AgdaSymbol{\AgdaUnderscore{})))}\AgdaSpace{}%
\AgdaBound{p}\<%
\\
\>[21]\AgdaOperator{\AgdaFunction{≡≡[}}\AgdaSpace{}%
\AgdaFunction{cong}\AgdaSpace{}%
\AgdaSymbol{(λ}\AgdaSpace{}%
\AgdaBound{e}\AgdaSpace{}%
\AgdaSymbol{→}\AgdaSpace{}%
\AgdaFunction{coe}\AgdaSpace{}%
\AgdaBound{e}\AgdaSpace{}%
\AgdaBound{p}\AgdaSymbol{)}\AgdaSpace{}%
\AgdaSymbol{(}\AgdaFunction{uip}\AgdaSpace{}%
\AgdaSymbol{\AgdaUnderscore{}}\AgdaSpace{}%
\AgdaSymbol{\AgdaUnderscore{})}\AgdaSpace{}%
\AgdaOperator{\AgdaFunction{]}}\<%
\\
\>[21][@{}l@{\AgdaIndent{0}}]%
\>[23]\AgdaFunction{coe}\AgdaSpace{}%
\AgdaSymbol{(}\AgdaPostulate{rfl}\AgdaSpace{}%
\AgdaSymbol{\AgdaUnderscore{})}\AgdaSpace{}%
\AgdaBound{p}\<%
\\
\>[21]\AgdaOperator{\AgdaFunction{≡≡[}}\AgdaSpace{}%
\AgdaFunction{coeIsRegular}\AgdaSpace{}%
\AgdaSymbol{\AgdaUnderscore{}}\AgdaSpace{}%
\AgdaBound{p}\AgdaSpace{}%
\AgdaOperator{\AgdaFunction{]}}\<%
\\
\>[21][@{}l@{\AgdaIndent{0}}]%
\>[23]\AgdaBound{p}\<%
\\
\>[21]\AgdaOperator{\AgdaFunction{qed}}\<%
\end{code}
\end{proof}

The elimination and computation properties of $\eq$ and $\heq$ then
follow: 

\begin{theorem}[Elimination and typal computation properties]
  \label{thm:elitcp}
  The axioms in Fig.~\ref{fig:axihtpe} imply that $\eq$ has the usual
  elimination and (typal) computation properties of homogeneous
  equality (in the form suggested by \cite{PaulinMohringC:inddsc})
  {\normalfont
\AgdaTarget{≡Elim}
\AgdaTarget{≡Comp}
\begin{code}%
\>[2]\AgdaFunction{≡Elim}\AgdaSpace{}%
\AgdaSymbol{:}%
\>[937I]\AgdaSymbol{∀\{}\AgdaBound{l}\AgdaSpace{}%
\AgdaBound{m}\AgdaSymbol{\}\{}\AgdaBound{A}\AgdaSpace{}%
\AgdaSymbol{:}\AgdaSpace{}%
\AgdaPrimitiveType{Set}\AgdaSpace{}%
\AgdaBound{l}\AgdaSymbol{\}\{}\AgdaBound{x}\AgdaSpace{}%
\AgdaSymbol{:}\AgdaSpace{}%
\AgdaBound{A}\AgdaSymbol{\}(}\AgdaBound{P}\AgdaSpace{}%
\AgdaSymbol{:}\AgdaSpace{}%
\AgdaSymbol{(}\AgdaBound{y}\AgdaSpace{}%
\AgdaSymbol{:}\AgdaSpace{}%
\AgdaBound{A}\AgdaSymbol{)}\AgdaSpace{}%
\AgdaSymbol{→}\AgdaSpace{}%
\AgdaBound{x}\AgdaSpace{}%
\AgdaOperator{\AgdaFunction{≡}}\AgdaSpace{}%
\AgdaBound{y}\AgdaSpace{}%
\AgdaSymbol{→}\AgdaSpace{}%
\AgdaPrimitiveType{Set}\AgdaSpace{}%
\AgdaBound{m}\AgdaSymbol{)}\<%
\\
\>[.][@{}l@{}]\<[937I]%
\>[10]\AgdaSymbol{(}\AgdaBound{p}\AgdaSpace{}%
\AgdaSymbol{:}\AgdaSpace{}%
\AgdaBound{P}\AgdaSpace{}%
\AgdaBound{x}\AgdaSpace{}%
\AgdaSymbol{(}\AgdaPostulate{rfl}\AgdaSpace{}%
\AgdaBound{x}\AgdaSymbol{))(}\AgdaBound{y}\AgdaSpace{}%
\AgdaSymbol{:}\AgdaSpace{}%
\AgdaBound{A}\AgdaSymbol{)(}\AgdaBound{e}\AgdaSpace{}%
\AgdaSymbol{:}\AgdaSpace{}%
\AgdaBound{x}\AgdaSpace{}%
\AgdaOperator{\AgdaFunction{≡}}\AgdaSpace{}%
\AgdaBound{y}\AgdaSymbol{)}\AgdaSpace{}%
\AgdaSymbol{→}\AgdaSpace{}%
\AgdaBound{P}\AgdaSpace{}%
\AgdaBound{y}\AgdaSpace{}%
\AgdaBound{e}\<%
\\
\>[2]\AgdaFunction{≡Comp}\AgdaSpace{}%
\AgdaSymbol{:}%
\>[971I]\AgdaSymbol{∀\{}\AgdaBound{l}\AgdaSpace{}%
\AgdaBound{m}\AgdaSymbol{\}\{}\AgdaBound{A}\AgdaSpace{}%
\AgdaSymbol{:}\AgdaSpace{}%
\AgdaPrimitiveType{Set}\AgdaSpace{}%
\AgdaBound{l}\AgdaSymbol{\}\{}\AgdaBound{x}\AgdaSpace{}%
\AgdaSymbol{:}\AgdaSpace{}%
\AgdaBound{A}\AgdaSymbol{\}(}\AgdaBound{P}\AgdaSpace{}%
\AgdaSymbol{:}\AgdaSpace{}%
\AgdaSymbol{(}\AgdaBound{y}\AgdaSpace{}%
\AgdaSymbol{:}\AgdaSpace{}%
\AgdaBound{A}\AgdaSymbol{)}\AgdaSpace{}%
\AgdaSymbol{→}\AgdaSpace{}%
\AgdaBound{x}\AgdaSpace{}%
\AgdaOperator{\AgdaFunction{≡}}\AgdaSpace{}%
\AgdaBound{y}\AgdaSpace{}%
\AgdaSymbol{→}\AgdaSpace{}%
\AgdaPrimitiveType{Set}\AgdaSpace{}%
\AgdaBound{m}\AgdaSymbol{)}\<%
\\
\>[.][@{}l@{}]\<[971I]%
\>[10]\AgdaSymbol{(}\AgdaBound{p}\AgdaSpace{}%
\AgdaSymbol{:}\AgdaSpace{}%
\AgdaBound{P}\AgdaSpace{}%
\AgdaBound{x}\AgdaSpace{}%
\AgdaSymbol{(}\AgdaPostulate{rfl}\AgdaSpace{}%
\AgdaBound{x}\AgdaSymbol{))}\AgdaSpace{}%
\AgdaSymbol{→}\AgdaSpace{}%
\AgdaFunction{≡Elim}\AgdaSpace{}%
\AgdaBound{P}\AgdaSpace{}%
\AgdaBound{p}\AgdaSpace{}%
\AgdaBound{x}\AgdaSpace{}%
\AgdaSymbol{(}\AgdaPostulate{rfl}\AgdaSpace{}%
\AgdaBound{x}\AgdaSymbol{)}\AgdaSpace{}%
\AgdaOperator{\AgdaFunction{≡}}\AgdaSpace{}%
\AgdaBound{p}\<%
\end{code}}%
  The axioms also imply that $\heq$ has the elimination and (typal)
  computation properties of heterogeneous equality described by
  \citet[Section~5.1.3]{McBrideC:deptfp}
  {\normalfont
\AgdaTarget{≡≡Elim}
\AgdaTarget{≡≡Comp}
\begin{code}%
\>[2]\AgdaFunction{≡≡Elim}%
\>[10]\AgdaSymbol{:}%
\>[1003I]\AgdaSymbol{∀\{}\AgdaBound{l}\AgdaSpace{}%
\AgdaBound{m}\AgdaSymbol{\}\{}\AgdaBound{A}\AgdaSpace{}%
\AgdaSymbol{:}\AgdaSpace{}%
\AgdaPrimitiveType{Set}\AgdaSpace{}%
\AgdaBound{l}\AgdaSymbol{\}\{}\AgdaBound{x}\AgdaSpace{}%
\AgdaSymbol{:}\AgdaSpace{}%
\AgdaBound{A}\AgdaSymbol{\}(}\AgdaBound{P}\AgdaSpace{}%
\AgdaSymbol{:}\AgdaSpace{}%
\AgdaSymbol{(}\AgdaBound{B}\AgdaSpace{}%
\AgdaSymbol{:}\AgdaSpace{}%
\AgdaPrimitiveType{Set}\AgdaSpace{}%
\AgdaBound{l}\AgdaSymbol{)(}\AgdaBound{y}\AgdaSpace{}%
\AgdaSymbol{:}\AgdaSpace{}%
\AgdaBound{B}\AgdaSymbol{)}\AgdaSpace{}%
\AgdaSymbol{→}\AgdaSpace{}%
\AgdaBound{x}\AgdaSpace{}%
\AgdaOperator{\AgdaPostulate{≡≡}}\AgdaSpace{}%
\AgdaBound{y}\AgdaSpace{}%
\AgdaSymbol{→}\AgdaSpace{}%
\AgdaPrimitiveType{Set}\AgdaSpace{}%
\AgdaBound{m}\AgdaSymbol{)}\<%
\\
\>[.][@{}l@{}]\<[1003I]%
\>[12]\AgdaSymbol{(}\AgdaBound{p}\AgdaSpace{}%
\AgdaSymbol{:}\AgdaSpace{}%
\AgdaBound{P}\AgdaSpace{}%
\AgdaBound{A}\AgdaSpace{}%
\AgdaBound{x}\AgdaSpace{}%
\AgdaSymbol{(}\AgdaPostulate{rfl}\AgdaSpace{}%
\AgdaBound{x}\AgdaSymbol{))(}\AgdaBound{B}\AgdaSpace{}%
\AgdaSymbol{:}\AgdaSpace{}%
\AgdaPrimitiveType{Set}\AgdaSpace{}%
\AgdaBound{l}\AgdaSymbol{)(}\AgdaBound{y}\AgdaSpace{}%
\AgdaSymbol{:}\AgdaSpace{}%
\AgdaBound{B}\AgdaSymbol{)(}\AgdaBound{e}\AgdaSpace{}%
\AgdaSymbol{:}\AgdaSpace{}%
\AgdaBound{x}\AgdaSpace{}%
\AgdaOperator{\AgdaPostulate{≡≡}}\AgdaSpace{}%
\AgdaBound{y}\AgdaSymbol{)}\AgdaSpace{}%
\AgdaSymbol{→}\AgdaSpace{}%
\AgdaBound{P}\AgdaSpace{}%
\AgdaBound{B}\AgdaSpace{}%
\AgdaBound{y}\AgdaSpace{}%
\AgdaBound{e}\<%
\\
\>[0]\<%
\\
\>[2]\AgdaFunction{≡≡Comp}%
\>[10]\AgdaSymbol{:}%
\>[1044I]\AgdaSymbol{∀\{}\AgdaBound{l}\AgdaSpace{}%
\AgdaBound{m}\AgdaSymbol{\}\{}\AgdaBound{A}\AgdaSpace{}%
\AgdaSymbol{:}\AgdaSpace{}%
\AgdaPrimitiveType{Set}\AgdaSpace{}%
\AgdaBound{l}\AgdaSymbol{\}\{}\AgdaBound{x}\AgdaSpace{}%
\AgdaSymbol{:}\AgdaSpace{}%
\AgdaBound{A}\AgdaSymbol{\}(}\AgdaBound{P}\AgdaSpace{}%
\AgdaSymbol{:}\AgdaSpace{}%
\AgdaSymbol{(}\AgdaBound{B}\AgdaSpace{}%
\AgdaSymbol{:}\AgdaSpace{}%
\AgdaPrimitiveType{Set}\AgdaSpace{}%
\AgdaBound{l}\AgdaSymbol{)(}\AgdaBound{y}\AgdaSpace{}%
\AgdaSymbol{:}\AgdaSpace{}%
\AgdaBound{B}\AgdaSymbol{)}\AgdaSpace{}%
\AgdaSymbol{→}\AgdaSpace{}%
\AgdaBound{x}\AgdaSpace{}%
\AgdaOperator{\AgdaPostulate{≡≡}}\AgdaSpace{}%
\AgdaBound{y}\AgdaSpace{}%
\AgdaSymbol{→}\AgdaSpace{}%
\AgdaPrimitiveType{Set}\AgdaSpace{}%
\AgdaBound{m}\AgdaSymbol{)}\<%
\\
\>[.][@{}l@{}]\<[1044I]%
\>[12]\AgdaSymbol{(}\AgdaBound{p}\AgdaSpace{}%
\AgdaSymbol{:}\AgdaSpace{}%
\AgdaBound{P}\AgdaSpace{}%
\AgdaBound{A}\AgdaSpace{}%
\AgdaBound{x}\AgdaSpace{}%
\AgdaSymbol{(}\AgdaPostulate{rfl}\AgdaSpace{}%
\AgdaBound{x}\AgdaSymbol{))}\AgdaSpace{}%
\AgdaSymbol{→}\AgdaSpace{}%
\AgdaFunction{≡≡Elim}\AgdaSpace{}%
\AgdaBound{P}\AgdaSpace{}%
\AgdaBound{p}\AgdaSpace{}%
\AgdaBound{A}\AgdaSpace{}%
\AgdaBound{x}\AgdaSpace{}%
\AgdaSymbol{(}\AgdaPostulate{rfl}\AgdaSpace{}%
\AgdaBound{x}\AgdaSymbol{)}\AgdaSpace{}%
\AgdaOperator{\AgdaFunction{≡}}\AgdaSpace{}%
\AgdaBound{p}\<%
\end{code}} 
\end{theorem}
\begin{proof}
  Using the functions from Fig~\ref{fig:equrheq} and Lemma~\ref{lem:verlt} we have:
\begin{code}%
\>[2]\AgdaFunction{≡Elim}\AgdaSpace{}%
\AgdaBound{P}\AgdaSpace{}%
\AgdaBound{p}\AgdaSpace{}%
\AgdaSymbol{\AgdaUnderscore{}}\AgdaSpace{}%
\AgdaBound{e}%
\>[21]\AgdaSymbol{=}\AgdaSpace{}%
\AgdaFunction{coe}\AgdaSpace{}%
\AgdaSymbol{(}\AgdaFunction{cong₂}\AgdaSpace{}%
\AgdaBound{P}\AgdaSpace{}%
\AgdaBound{e}\AgdaSpace{}%
\AgdaSymbol{(}\AgdaPostulate{ctr}\AgdaSpace{}%
\AgdaBound{e}\AgdaSymbol{))}\AgdaSpace{}%
\AgdaBound{p}\<%
\\
\\[\AgdaEmptyExtraSkip]%
\>[2]\AgdaFunction{≡Comp}\AgdaSpace{}%
\AgdaBound{P}\AgdaSpace{}%
\AgdaBound{p}%
\>[21]\AgdaSymbol{=}%
\>[24]\AgdaOperator{\AgdaFunction{proof}}\<%
\\
\>[24][@{}l@{\AgdaIndent{0}}]%
\>[26]\AgdaFunction{coe}\AgdaSpace{}%
\AgdaSymbol{(}\AgdaFunction{cong₂}\AgdaSpace{}%
\AgdaBound{P}\AgdaSpace{}%
\AgdaSymbol{(}\AgdaPostulate{rfl}\AgdaSpace{}%
\AgdaSymbol{\AgdaUnderscore{})}\AgdaSpace{}%
\AgdaSymbol{(}\AgdaPostulate{ctr}\AgdaSpace{}%
\AgdaSymbol{(}\AgdaPostulate{rfl}\AgdaSpace{}%
\AgdaSymbol{\AgdaUnderscore{})))}\AgdaSpace{}%
\AgdaBound{p}\<%
\\
\>[24]\AgdaOperator{\AgdaFunction{≡≡[}}\AgdaSpace{}%
\AgdaFunction{cong}\AgdaSpace{}%
\AgdaSymbol{(λ}\AgdaSpace{}%
\AgdaBound{e}\AgdaSpace{}%
\AgdaSymbol{→}\AgdaSpace{}%
\AgdaFunction{coe}\AgdaSpace{}%
\AgdaBound{e}\AgdaSpace{}%
\AgdaBound{p}\AgdaSymbol{)}\AgdaSpace{}%
\AgdaSymbol{(}\AgdaFunction{symm}\AgdaSpace{}%
\AgdaSymbol{(}\AgdaPostulate{ctr}\AgdaSpace{}%
\AgdaSymbol{\AgdaUnderscore{}))}\AgdaSpace{}%
\AgdaOperator{\AgdaFunction{]}}\<%
\\
\>[24][@{}l@{\AgdaIndent{0}}]%
\>[26]\AgdaFunction{coe}\AgdaSpace{}%
\AgdaSymbol{(}\AgdaPostulate{rfl}\AgdaSpace{}%
\AgdaSymbol{\AgdaUnderscore{})}\AgdaSpace{}%
\AgdaBound{p}\<%
\\
\>[24]\AgdaOperator{\AgdaFunction{≡≡[}}\AgdaSpace{}%
\AgdaFunction{coeIsRegular}\AgdaSpace{}%
\AgdaSymbol{\AgdaUnderscore{}}\AgdaSpace{}%
\AgdaBound{p}\AgdaSpace{}%
\AgdaOperator{\AgdaFunction{]}}\<%
\\
\>[24][@{}l@{\AgdaIndent{0}}]%
\>[26]\AgdaBound{p}\<%
\\
\>[24]\AgdaOperator{\AgdaFunction{qed}}\<%
\\
\\[\AgdaEmptyExtraSkip]%
\>[2]\AgdaFunction{≡≡Elim}\AgdaSpace{}%
\AgdaBound{P}\AgdaSpace{}%
\AgdaBound{p}\AgdaSpace{}%
\AgdaSymbol{\AgdaUnderscore{}}\AgdaSpace{}%
\AgdaSymbol{\AgdaUnderscore{}}\AgdaSpace{}%
\AgdaBound{e}%
\>[21]\AgdaSymbol{=}\AgdaSpace{}%
\AgdaFunction{coe}\AgdaSpace{}%
\AgdaSymbol{(}\AgdaFunction{cong₃}\AgdaSpace{}%
\AgdaBound{P}\AgdaSpace{}%
\AgdaSymbol{(}\AgdaPostulate{eqt}\AgdaSpace{}%
\AgdaBound{e}\AgdaSymbol{)}\AgdaSpace{}%
\AgdaBound{e}\AgdaSpace{}%
\AgdaSymbol{(}\AgdaPostulate{ctr}\AgdaSpace{}%
\AgdaBound{e}\AgdaSymbol{))}\AgdaSpace{}%
\AgdaBound{p}\<%
\\
\\[\AgdaEmptyExtraSkip]%
\>[2]\AgdaFunction{≡≡Comp}\AgdaSpace{}%
\AgdaBound{P}\AgdaSpace{}%
\AgdaBound{p}%
\>[21]\AgdaSymbol{=}%
\>[24]\AgdaOperator{\AgdaFunction{proof}}\<%
\\
\>[24][@{}l@{\AgdaIndent{0}}]%
\>[26]\AgdaFunction{coe}\AgdaSpace{}%
\AgdaSymbol{(}\AgdaFunction{cong₃}\AgdaSpace{}%
\AgdaBound{P}\AgdaSpace{}%
\AgdaSymbol{(}\AgdaPostulate{eqt}\AgdaSpace{}%
\AgdaSymbol{(}\AgdaPostulate{rfl}\AgdaSpace{}%
\AgdaSymbol{\AgdaUnderscore{}))}\AgdaSpace{}%
\AgdaSymbol{(}\AgdaPostulate{rfl}\AgdaSpace{}%
\AgdaSymbol{\AgdaUnderscore{})}\AgdaSpace{}%
\AgdaSymbol{(}\AgdaPostulate{ctr}\AgdaSpace{}%
\AgdaSymbol{(}\AgdaPostulate{rfl}\AgdaSpace{}%
\AgdaSymbol{\AgdaUnderscore{})))}\AgdaSpace{}%
\AgdaBound{p}\<%
\\
\>[24]\AgdaOperator{\AgdaFunction{≡≡[}}\AgdaSpace{}%
\AgdaFunction{cong}\AgdaSpace{}%
\AgdaSymbol{(λ}\AgdaSpace{}%
\AgdaBound{e}\AgdaSpace{}%
\AgdaSymbol{→}\AgdaSpace{}%
\AgdaFunction{coe}\AgdaSpace{}%
\AgdaBound{e}\AgdaSpace{}%
\AgdaBound{p}\AgdaSymbol{)}\AgdaSpace{}%
\AgdaSymbol{(}\AgdaFunction{uip}\AgdaSpace{}%
\AgdaSymbol{\AgdaUnderscore{}}\AgdaSpace{}%
\AgdaSymbol{\AgdaUnderscore{})}\AgdaSpace{}%
\AgdaOperator{\AgdaFunction{]}}\<%
\\
\>[24][@{}l@{\AgdaIndent{0}}]%
\>[26]\AgdaFunction{coe}\AgdaSpace{}%
\AgdaSymbol{(}\AgdaPostulate{rfl}\AgdaSpace{}%
\AgdaSymbol{\AgdaUnderscore{})}\AgdaSpace{}%
\AgdaBound{p}\<%
\\
\>[24]\AgdaOperator{\AgdaFunction{≡≡[}}\AgdaSpace{}%
\AgdaFunction{coeIsRegular}\AgdaSpace{}%
\AgdaSymbol{\AgdaUnderscore{}}\AgdaSpace{}%
\AgdaBound{p}\AgdaSpace{}%
\AgdaOperator{\AgdaFunction{]}}\<%
\\
\>[24][@{}l@{\AgdaIndent{0}}]%
\>[26]\AgdaBound{p}\<%
\\
\>[21][@{}l@{\AgdaIndent{0}}]%
\>[23]\AgdaOperator{\AgdaFunction{qed}}\<%
\end{code}
\end{proof}

Note that a corollary of the above two theorems is that $\_\heq\_$ is
uniquely determined up to logical equivalence by the axioms in
Fig.~\ref{fig:axihtpe}. In other words, for any other such family of
types $\_\heq'\_\,$, there are functions in either direction between
$\aarg{x}\heq\aarg{y}$ and $\aarg{x}\heq'\aarg{y}\,$; and because of UIP
these are necessarily mutually inverse up to $\heq$ (or $\heq'$).

\begin{remark}
  $\afun{≡≡Elim}$ is the elimination form systematically
  derived~\citep{BackhouseR:dott} from $\_\heq\_$ and $\afun{rfl}$,
  regarding them as the formation and introduction rules for an
  inductive type. As \citet[page~120]{McBrideC:deptfp} points out,
  $\afun{≡≡Elim}$ is not very useful, because of the way it's motive
  $\aarg{P}$ involves abstraction over an arbitrary type
  $\aarg{B}$. McBride goes on to give another, more useful form of
  elimination for $\heq$, but in our setting where $\eq$ is a special
  case of $\heq$, that coincides with the eliminator $\afun{≡Elim}$.
\end{remark}

\begin{remark}[The role of $\Sigma$-types]
  \label{rem:rolst}
  One of the strengths of machine-checked mathematics is that it aids
  the detection of logical dependency. Although we included the
  equations $\afun{fpr}$, $\afun{spr}$ and $\afun{eta}$ for
  $\Sigma$-types in Fig.~\ref{fig:axihtpe}, they have not been used
  for the results so far, as may be verified by commenting them out
  from this literate Agda file and re-checking it up to this point.

  So only the weak form of dependent product given by $\afun{∑}$,
  $\_\afun{,}\_$, $\afun{fst}$ and $\afun{snd}$ in the figure is used
  to define the regular version of coercion in Lemma~\ref{lem:verlt}
  and then prove Theorems~\ref{thm:uipak} and \ref{thm:elitcp}. It
  would be nice if there was some way to define $\afun{∑}$,
  $\_\afun{,}\_$, $\afun{fst}$ and $\afun{snd}$ just using dependent
  function types and universes.

  However, the extra equations $\afun{fpr}$, $\afun{spr}$ and
  $\afun{eta}$ for $\afun{∑}$ are of course very natural. Let us
  record the fact that they enable one to define the usual elimination
  rule for dependent products, with a typal computation rule:
  \AgdaTarget{∑Elim} \AgdaTarget{∑Comp}
\begin{code}%
\>[2]\AgdaFunction{∑Elim}%
\>[17]\AgdaSymbol{:}%
\>[20]\AgdaSymbol{∀\{}\AgdaBound{l}\AgdaSpace{}%
\AgdaBound{m}\AgdaSpace{}%
\AgdaBound{n}\AgdaSymbol{\}\{}\AgdaBound{A}\AgdaSpace{}%
\AgdaSymbol{:}\AgdaSpace{}%
\AgdaPrimitiveType{Set}\AgdaSpace{}%
\AgdaBound{l}\AgdaSymbol{\}\{}\AgdaBound{B}\AgdaSpace{}%
\AgdaSymbol{:}\AgdaSpace{}%
\AgdaBound{A}\AgdaSpace{}%
\AgdaSymbol{→}\AgdaSpace{}%
\AgdaPrimitiveType{Set}\AgdaSpace{}%
\AgdaBound{m}\AgdaSymbol{\}(}\AgdaBound{C}\AgdaSpace{}%
\AgdaSymbol{:}\AgdaSpace{}%
\AgdaPostulate{∑}\AgdaSpace{}%
\AgdaBound{A}\AgdaSpace{}%
\AgdaBound{B}\AgdaSpace{}%
\AgdaSymbol{→}\AgdaSpace{}%
\AgdaPrimitiveType{Set}\AgdaSpace{}%
\AgdaBound{n}\AgdaSymbol{)}\<%
\\
\>[20]\AgdaSymbol{(}\AgdaBound{c}\AgdaSpace{}%
\AgdaSymbol{:}\AgdaSpace{}%
\AgdaSymbol{(}\AgdaBound{x}\AgdaSpace{}%
\AgdaSymbol{:}\AgdaSpace{}%
\AgdaBound{A}\AgdaSymbol{)(}\AgdaBound{y}\AgdaSpace{}%
\AgdaSymbol{:}\AgdaSpace{}%
\AgdaBound{B}\AgdaSpace{}%
\AgdaBound{x}\AgdaSymbol{)}\AgdaSpace{}%
\AgdaSymbol{→}\AgdaSpace{}%
\AgdaBound{C}\AgdaSpace{}%
\AgdaSymbol{(}\AgdaBound{x}\AgdaSpace{}%
\AgdaOperator{\AgdaPostulate{,}}\AgdaSpace{}%
\AgdaBound{y}\AgdaSymbol{))(}\AgdaBound{z}\AgdaSpace{}%
\AgdaSymbol{:}\AgdaSpace{}%
\AgdaPostulate{∑}\AgdaSpace{}%
\AgdaBound{A}\AgdaSpace{}%
\AgdaBound{B}\AgdaSymbol{)}\AgdaSpace{}%
\AgdaSymbol{→}\AgdaSpace{}%
\AgdaBound{C}\AgdaSpace{}%
\AgdaBound{z}\<%
\\
\>[2]\AgdaFunction{∑Elim}\AgdaSpace{}%
\AgdaBound{C}\AgdaSpace{}%
\AgdaBound{c}\AgdaSpace{}%
\AgdaBound{z}%
\>[17]\AgdaSymbol{=}%
\>[20]\AgdaFunction{coe}\AgdaSpace{}%
\AgdaSymbol{(}\AgdaFunction{cong}\AgdaSpace{}%
\AgdaBound{C}\AgdaSpace{}%
\AgdaSymbol{(}\AgdaPostulate{eta}\AgdaSpace{}%
\AgdaBound{z}\AgdaSymbol{))}\AgdaSpace{}%
\AgdaSymbol{(}\AgdaBound{c}\AgdaSpace{}%
\AgdaSymbol{(}\AgdaPostulate{fst}\AgdaSpace{}%
\AgdaBound{z}\AgdaSymbol{)}\AgdaSpace{}%
\AgdaSymbol{(}\AgdaPostulate{snd}\AgdaSpace{}%
\AgdaBound{z}\AgdaSymbol{))}\<%
\\
\>[0]\<%
\\
\>[2]\AgdaFunction{∑Comp}%
\>[17]\AgdaSymbol{:}%
\>[20]\AgdaSymbol{∀}\AgdaSpace{}%
\AgdaSymbol{\{}\AgdaBound{l}\AgdaSpace{}%
\AgdaBound{m}\AgdaSpace{}%
\AgdaBound{n}\AgdaSymbol{\}\{}\AgdaBound{A}\AgdaSpace{}%
\AgdaSymbol{:}\AgdaSpace{}%
\AgdaPrimitiveType{Set}\AgdaSpace{}%
\AgdaBound{l}\AgdaSymbol{\}\{}\AgdaBound{B}\AgdaSpace{}%
\AgdaSymbol{:}\AgdaSpace{}%
\AgdaBound{A}\AgdaSpace{}%
\AgdaSymbol{→}\AgdaSpace{}%
\AgdaPrimitiveType{Set}\AgdaSpace{}%
\AgdaBound{m}\AgdaSymbol{\}(}\AgdaBound{C}\AgdaSpace{}%
\AgdaSymbol{:}\AgdaSpace{}%
\AgdaPostulate{∑}\AgdaSpace{}%
\AgdaBound{A}\AgdaSpace{}%
\AgdaBound{B}\AgdaSpace{}%
\AgdaSymbol{→}\AgdaSpace{}%
\AgdaPrimitiveType{Set}\AgdaSpace{}%
\AgdaBound{n}\AgdaSymbol{)}\<%
\\
\>[20]\AgdaSymbol{(}\AgdaBound{c}\AgdaSpace{}%
\AgdaSymbol{:}\AgdaSpace{}%
\AgdaSymbol{(}\AgdaBound{x}\AgdaSpace{}%
\AgdaSymbol{:}\AgdaSpace{}%
\AgdaBound{A}\AgdaSymbol{)(}\AgdaBound{y}\AgdaSpace{}%
\AgdaSymbol{:}\AgdaSpace{}%
\AgdaBound{B}\AgdaSpace{}%
\AgdaBound{x}\AgdaSymbol{)}\AgdaSpace{}%
\AgdaSymbol{→}\AgdaSpace{}%
\AgdaBound{C}\AgdaSpace{}%
\AgdaSymbol{(}\AgdaBound{x}\AgdaSpace{}%
\AgdaOperator{\AgdaPostulate{,}}\AgdaSpace{}%
\AgdaBound{y}\AgdaSymbol{))(}\AgdaBound{x}\AgdaSpace{}%
\AgdaSymbol{:}\AgdaSpace{}%
\AgdaBound{A}\AgdaSymbol{)(}\AgdaBound{y}\AgdaSpace{}%
\AgdaSymbol{:}\AgdaSpace{}%
\AgdaBound{B}\AgdaSpace{}%
\AgdaBound{x}\AgdaSymbol{)}\AgdaSpace{}%
\AgdaSymbol{→}\<%
\\
\>[20]\AgdaFunction{∑Elim}\AgdaSpace{}%
\AgdaBound{C}\AgdaSpace{}%
\AgdaBound{c}\AgdaSpace{}%
\AgdaSymbol{(}\AgdaBound{x}\AgdaSpace{}%
\AgdaOperator{\AgdaPostulate{,}}\AgdaSpace{}%
\AgdaBound{y}\AgdaSymbol{)}\AgdaSpace{}%
\AgdaOperator{\AgdaFunction{≡}}\AgdaSpace{}%
\AgdaBound{c}\AgdaSpace{}%
\AgdaBound{x}\AgdaSpace{}%
\AgdaBound{y}\<%
\\
\>[2]\AgdaFunction{∑Comp}\AgdaSpace{}%
\AgdaBound{C}\AgdaSpace{}%
\AgdaBound{c}\AgdaSpace{}%
\AgdaBound{x}\AgdaSpace{}%
\AgdaBound{y}%
\>[17]\AgdaSymbol{=}%
\>[20]\AgdaKeyword{let}\AgdaSpace{}%
\AgdaBound{z}\AgdaSpace{}%
\AgdaSymbol{=}\AgdaSpace{}%
\AgdaSymbol{(}\AgdaBound{x}\AgdaSpace{}%
\AgdaOperator{\AgdaPostulate{,}}\AgdaSpace{}%
\AgdaBound{y}\AgdaSymbol{)}\AgdaSpace{}%
\AgdaKeyword{in}\<%
\\
\>[20]\AgdaOperator{\AgdaFunction{proof}}\<%
\\
\>[20][@{}l@{\AgdaIndent{0}}]%
\>[22]\AgdaFunction{coe}\AgdaSpace{}%
\AgdaSymbol{(}\AgdaFunction{cong}\AgdaSpace{}%
\AgdaBound{C}\AgdaSpace{}%
\AgdaSymbol{(}\AgdaPostulate{eta}\AgdaSpace{}%
\AgdaBound{z}\AgdaSymbol{))}\AgdaSpace{}%
\AgdaSymbol{(}\AgdaBound{c}\AgdaSpace{}%
\AgdaSymbol{(}\AgdaPostulate{fst}\AgdaSpace{}%
\AgdaBound{z}\AgdaSymbol{)}\AgdaSpace{}%
\AgdaSymbol{(}\AgdaPostulate{snd}\AgdaSpace{}%
\AgdaBound{z}\AgdaSymbol{))}\<%
\\
\>[20]\AgdaOperator{\AgdaFunction{≡≡[}}\AgdaSpace{}%
\AgdaFunction{coeIsRegular}\AgdaSpace{}%
\AgdaSymbol{\AgdaUnderscore{}}\AgdaSpace{}%
\AgdaSymbol{\AgdaUnderscore{}}\AgdaSpace{}%
\AgdaOperator{\AgdaFunction{]}}\<%
\\
\>[20][@{}l@{\AgdaIndent{0}}]%
\>[22]\AgdaBound{c}\AgdaSpace{}%
\AgdaSymbol{(}\AgdaPostulate{fst}\AgdaSpace{}%
\AgdaBound{z}\AgdaSymbol{)}\AgdaSpace{}%
\AgdaSymbol{(}\AgdaPostulate{snd}\AgdaSpace{}%
\AgdaBound{z}\AgdaSymbol{)}\<%
\\
\>[20]\AgdaOperator{\AgdaFunction{≡≡[}}\AgdaSpace{}%
\AgdaFunction{cong₂}\AgdaSpace{}%
\AgdaBound{c}\AgdaSpace{}%
\AgdaSymbol{(}\AgdaPostulate{fpr}\AgdaSpace{}%
\AgdaBound{x}\AgdaSpace{}%
\AgdaBound{y}\AgdaSymbol{)}\AgdaSpace{}%
\AgdaSymbol{(}\AgdaPostulate{spr}\AgdaSpace{}%
\AgdaBound{x}\AgdaSpace{}%
\AgdaBound{y}\AgdaSymbol{)}\AgdaSpace{}%
\AgdaOperator{\AgdaFunction{]}}\<%
\\
\>[20][@{}l@{\AgdaIndent{0}}]%
\>[22]\AgdaBound{c}\AgdaSpace{}%
\AgdaBound{x}\AgdaSpace{}%
\AgdaBound{y}\<%
\\
\>[20]\AgdaOperator{\AgdaFunction{qed}}\<%
\end{code}
\end{remark}

\section{Consistency of the axioms}
\label{sec:cona}

We have seen that the axioms in Fig.~\ref{fig:axihtpe} suffice to
define dependent products and both heterogeneous and homogeneous
equality types with uniqueness of identity proofs, all satisfying the
usual elimination properties, albeit with typal computation
rules. Conversely it is not hard to see that the elimination
and computation rules in Theorem~\ref{thm:elitcp} and
Remark~\ref{rem:rolst}, together with Axiom~K, imply the axioms in
Fig.~\ref{fig:axihtpe}. Instead of doing that, in this section we just
check that the axioms are provable from inductive definitions of
equality and dependent product types.  One can make these inductive
definitions in Agda as follows:
\begin{code}[hide]%
\>[0]\AgdaKeyword{module}\AgdaSpace{}%
\AgdaModule{Consistency}\AgdaSpace{}%
\AgdaKeyword{where}\<%
\\
\>[0][@{}l@{\AgdaIndent{0}}]%
\>[2]\AgdaKeyword{infix}%
\>[9]\AgdaNumber{3}\AgdaSpace{}%
\AgdaOperator{\AgdaInductiveConstructor{\AgdaUnderscore{},\AgdaUnderscore{}}}\<%
\\
\>[2]\AgdaKeyword{infix}\AgdaSpace{}%
\AgdaNumber{4}\AgdaSpace{}%
\AgdaOperator{\AgdaDatatype{\AgdaUnderscore{}≡≡\AgdaUnderscore{}}}\AgdaSpace{}%
\AgdaOperator{\AgdaFunction{\AgdaUnderscore{}≡\AgdaUnderscore{}}}\<%
\\
\>[0]\<%
\\
\>[2]\AgdaComment{-- -- without this \AgdaUnderscore{}≡≡\AgdaUnderscore{} would have to land in Set (lsuc l)}\<%
\\
\>[2]\AgdaComment{-- \{-\# NO\AgdaUnderscore{}UNIVERSE\AgdaUnderscore{}CHECK \#-\}  }\<%
\end{code}
\begin{code}%
\>[2]\AgdaKeyword{data}\AgdaSpace{}%
\AgdaOperator{\AgdaDatatype{\AgdaUnderscore{}≡≡\AgdaUnderscore{}}}\AgdaSpace{}%
\AgdaSymbol{\{}\AgdaBound{l}\AgdaSymbol{\}\{}\AgdaBound{A}\AgdaSpace{}%
\AgdaSymbol{:}\AgdaSpace{}%
\AgdaPrimitiveType{Set}\AgdaSpace{}%
\AgdaBound{l}\AgdaSymbol{\}}\AgdaSpace{}%
\AgdaSymbol{:}\AgdaSpace{}%
\AgdaSymbol{\{}\AgdaBound{B}\AgdaSpace{}%
\AgdaSymbol{:}\AgdaSpace{}%
\AgdaPrimitiveType{Set}\AgdaSpace{}%
\AgdaBound{l}\AgdaSymbol{\}}\AgdaSpace{}%
\AgdaSymbol{→}\AgdaSpace{}%
\AgdaBound{A}\AgdaSpace{}%
\AgdaSymbol{→}\AgdaSpace{}%
\AgdaBound{B}\AgdaSpace{}%
\AgdaSymbol{→}\AgdaSpace{}%
\AgdaPrimitiveType{Set}\AgdaSpace{}%
\AgdaBound{l}\AgdaSpace{}%
\AgdaKeyword{where}\<%
\\
\>[2][@{}l@{\AgdaIndent{0}}]%
\>[4]\AgdaInductiveConstructor{rfl}\AgdaSpace{}%
\AgdaSymbol{:}\AgdaSpace{}%
\AgdaSymbol{(}\AgdaBound{x}\AgdaSpace{}%
\AgdaSymbol{:}\AgdaSpace{}%
\AgdaBound{A}\AgdaSymbol{)}\AgdaSpace{}%
\AgdaSymbol{→}\AgdaSpace{}%
\AgdaBound{x}\AgdaSpace{}%
\AgdaOperator{\AgdaDatatype{≡≡}}\AgdaSpace{}%
\AgdaBound{x}\<%
\\
\>[2]\AgdaKeyword{data}\AgdaSpace{}%
\AgdaDatatype{∑}\AgdaSpace{}%
\AgdaSymbol{\{}\AgdaBound{l}\AgdaSpace{}%
\AgdaBound{m}\AgdaSymbol{\}(}\AgdaBound{A}\AgdaSpace{}%
\AgdaSymbol{:}\AgdaSpace{}%
\AgdaPrimitiveType{Set}\AgdaSpace{}%
\AgdaBound{l}\AgdaSymbol{)(}\AgdaBound{B}\AgdaSpace{}%
\AgdaSymbol{:}\AgdaSpace{}%
\AgdaBound{A}\AgdaSpace{}%
\AgdaSymbol{→}\AgdaSpace{}%
\AgdaPrimitiveType{Set}\AgdaSpace{}%
\AgdaBound{m}\AgdaSymbol{)}\AgdaSpace{}%
\AgdaSymbol{:}\AgdaSpace{}%
\AgdaPrimitiveType{Set}\AgdaSpace{}%
\AgdaSymbol{(}\AgdaBound{l}\AgdaSpace{}%
\AgdaOperator{\AgdaPrimitive{⊔}}\AgdaSpace{}%
\AgdaBound{m}\AgdaSymbol{)}\AgdaSpace{}%
\AgdaKeyword{where}\<%
\\
\>[2][@{}l@{\AgdaIndent{0}}]%
\>[4]\AgdaOperator{\AgdaInductiveConstructor{\AgdaUnderscore{},\AgdaUnderscore{}}}\AgdaSpace{}%
\AgdaSymbol{:}\AgdaSpace{}%
\AgdaSymbol{(}\AgdaBound{x}\AgdaSpace{}%
\AgdaSymbol{:}\AgdaSpace{}%
\AgdaBound{A}\AgdaSymbol{)}\AgdaSpace{}%
\AgdaSymbol{→}\AgdaSpace{}%
\AgdaBound{B}\AgdaSpace{}%
\AgdaBound{x}\AgdaSpace{}%
\AgdaSymbol{→}\AgdaSpace{}%
\AgdaDatatype{∑}\AgdaSpace{}%
\AgdaBound{A}\AgdaSpace{}%
\AgdaBound{B}\<%
\\
\>[2]\AgdaComment{-- the derived homogeneous equality}\<%
\\
\>[2]\AgdaOperator{\AgdaFunction{\AgdaUnderscore{}≡\AgdaUnderscore{}}}\AgdaSpace{}%
\AgdaSymbol{:}\AgdaSpace{}%
\AgdaSymbol{∀\{}\AgdaBound{l}\AgdaSymbol{\}\{}\AgdaBound{A}\AgdaSpace{}%
\AgdaSymbol{:}\AgdaSpace{}%
\AgdaPrimitiveType{Set}\AgdaSpace{}%
\AgdaBound{l}\AgdaSymbol{\}}\AgdaSpace{}%
\AgdaSymbol{→}\AgdaSpace{}%
\AgdaBound{A}\AgdaSpace{}%
\AgdaSymbol{→}\AgdaSpace{}%
\AgdaBound{A}\AgdaSpace{}%
\AgdaSymbol{→}\AgdaSpace{}%
\AgdaPrimitiveType{Set}\AgdaSpace{}%
\AgdaBound{l}\<%
\\
\>[2]\AgdaBound{x}\AgdaSpace{}%
\AgdaOperator{\AgdaFunction{≡}}\AgdaSpace{}%
\AgdaBound{y}\AgdaSpace{}%
\AgdaSymbol{=}\AgdaSpace{}%
\AgdaBound{x}\AgdaSpace{}%
\AgdaOperator{\AgdaDatatype{≡≡}}\AgdaSpace{}%
\AgdaBound{y}\<%
\end{code}
Then Agda's implementation of dependent pattern matching enables
straightforward definitions of the functions from
Fig.~\ref{fig:axihtpe}, as follows:
\begin{code}%
\>[2]\AgdaFunction{ctr}\AgdaSpace{}%
\AgdaSymbol{:}\AgdaSpace{}%
\AgdaSymbol{∀\{}\AgdaBound{l}\AgdaSymbol{\}\{}\AgdaBound{A}\AgdaSpace{}%
\AgdaBound{B}\AgdaSpace{}%
\AgdaSymbol{:}\AgdaSpace{}%
\AgdaPrimitiveType{Set}\AgdaSpace{}%
\AgdaBound{l}\AgdaSymbol{\}\{}\AgdaBound{x}\AgdaSpace{}%
\AgdaSymbol{:}\AgdaSpace{}%
\AgdaBound{A}\AgdaSymbol{\}\{}\AgdaBound{y}\AgdaSpace{}%
\AgdaSymbol{:}\AgdaSpace{}%
\AgdaBound{B}\AgdaSymbol{\}(}\AgdaBound{e}\AgdaSpace{}%
\AgdaSymbol{:}\AgdaSpace{}%
\AgdaBound{x}\AgdaSpace{}%
\AgdaOperator{\AgdaDatatype{≡≡}}\AgdaSpace{}%
\AgdaBound{y}\AgdaSymbol{)}\AgdaSpace{}%
\AgdaSymbol{→}\AgdaSpace{}%
\AgdaInductiveConstructor{rfl}\AgdaSpace{}%
\AgdaBound{x}\AgdaSpace{}%
\AgdaOperator{\AgdaDatatype{≡≡}}\AgdaSpace{}%
\AgdaBound{e}\<%
\\
\>[2]\AgdaFunction{ctr}\AgdaSpace{}%
\AgdaSymbol{(}\AgdaInductiveConstructor{rfl}\AgdaSpace{}%
\AgdaBound{x}\AgdaSymbol{)}\AgdaSpace{}%
\AgdaSymbol{=}\AgdaSpace{}%
\AgdaInductiveConstructor{rfl}\AgdaSpace{}%
\AgdaSymbol{(}\AgdaInductiveConstructor{rfl}\AgdaSpace{}%
\AgdaBound{x}\AgdaSymbol{)}\<%
\\
\>[0]\<%
\\
\>[2]\AgdaFunction{eqt}\AgdaSpace{}%
\AgdaSymbol{:}\AgdaSpace{}%
\AgdaSymbol{∀\{}\AgdaBound{l}\AgdaSymbol{\}\{}\AgdaBound{A}\AgdaSpace{}%
\AgdaBound{B}\AgdaSpace{}%
\AgdaSymbol{:}\AgdaSpace{}%
\AgdaPrimitiveType{Set}\AgdaSpace{}%
\AgdaBound{l}\AgdaSymbol{\}\{}\AgdaBound{x}\AgdaSpace{}%
\AgdaSymbol{:}\AgdaSpace{}%
\AgdaBound{A}\AgdaSymbol{\}\{}\AgdaBound{y}\AgdaSpace{}%
\AgdaSymbol{:}\AgdaSpace{}%
\AgdaBound{B}\AgdaSymbol{\}}\AgdaSpace{}%
\AgdaSymbol{→}\AgdaSpace{}%
\AgdaBound{x}\AgdaSpace{}%
\AgdaOperator{\AgdaDatatype{≡≡}}\AgdaSpace{}%
\AgdaBound{y}\AgdaSpace{}%
\AgdaSymbol{→}\AgdaSpace{}%
\AgdaBound{A}\AgdaSpace{}%
\AgdaOperator{\AgdaFunction{≡}}\AgdaSpace{}%
\AgdaBound{B}\<%
\\
\>[2]\AgdaFunction{eqt}\AgdaSpace{}%
\AgdaSymbol{\{\AgdaUnderscore{}\}}\AgdaSpace{}%
\AgdaSymbol{\{}\AgdaBound{A}\AgdaSymbol{\}}\AgdaSpace{}%
\AgdaSymbol{(}\AgdaInductiveConstructor{rfl}\AgdaSpace{}%
\AgdaSymbol{\AgdaUnderscore{})}\AgdaSpace{}%
\AgdaSymbol{=}\AgdaSpace{}%
\AgdaInductiveConstructor{rfl}\AgdaSpace{}%
\AgdaBound{A}\<%
\\
\\[\AgdaEmptyExtraSkip]%
\>[2]\AgdaFunction{tpt}\AgdaSpace{}%
\AgdaSymbol{:}%
\>[1425I]\AgdaSymbol{∀\{}\AgdaBound{l}\AgdaSpace{}%
\AgdaBound{m}\AgdaSpace{}%
\AgdaBound{n}\AgdaSymbol{\}\{}\AgdaBound{A}\AgdaSpace{}%
\AgdaSymbol{:}\AgdaSpace{}%
\AgdaPrimitiveType{Set}\AgdaSpace{}%
\AgdaBound{l}\AgdaSymbol{\}\{}\AgdaBound{B}\AgdaSpace{}%
\AgdaSymbol{:}\AgdaSpace{}%
\AgdaBound{A}\AgdaSpace{}%
\AgdaSymbol{→}\AgdaSpace{}%
\AgdaPrimitiveType{Set}\AgdaSpace{}%
\AgdaBound{m}\AgdaSymbol{\}(}\AgdaBound{C}\AgdaSpace{}%
\AgdaSymbol{:}\AgdaSpace{}%
\AgdaSymbol{(}\AgdaBound{x}\AgdaSpace{}%
\AgdaSymbol{:}\AgdaSpace{}%
\AgdaBound{A}\AgdaSymbol{)}\AgdaSpace{}%
\AgdaSymbol{→}\AgdaSpace{}%
\AgdaBound{B}\AgdaSpace{}%
\AgdaBound{x}\AgdaSpace{}%
\AgdaSymbol{→}\AgdaSpace{}%
\AgdaPrimitiveType{Set}\AgdaSpace{}%
\AgdaBound{n}\AgdaSymbol{)}\<%
\\
\>[.][@{}l@{}]\<[1425I]%
\>[8]\AgdaSymbol{\{}\AgdaBound{x}\AgdaSpace{}%
\AgdaBound{x′}\AgdaSpace{}%
\AgdaSymbol{:}\AgdaSpace{}%
\AgdaBound{A}\AgdaSymbol{\}\{}\AgdaBound{y}\AgdaSpace{}%
\AgdaSymbol{:}\AgdaSpace{}%
\AgdaBound{B}\AgdaSpace{}%
\AgdaBound{x}\AgdaSymbol{\}\{}\AgdaBound{y′}\AgdaSpace{}%
\AgdaSymbol{:}\AgdaSpace{}%
\AgdaBound{B}\AgdaSpace{}%
\AgdaBound{x′}\AgdaSymbol{\}}\AgdaSpace{}%
\AgdaSymbol{→}\AgdaSpace{}%
\AgdaBound{x}\AgdaSpace{}%
\AgdaOperator{\AgdaFunction{≡}}\AgdaSpace{}%
\AgdaBound{x′}\AgdaSpace{}%
\AgdaSymbol{→}\AgdaSpace{}%
\AgdaBound{y}\AgdaSpace{}%
\AgdaOperator{\AgdaDatatype{≡≡}}\AgdaSpace{}%
\AgdaBound{y′}\AgdaSpace{}%
\AgdaSymbol{→}\AgdaSpace{}%
\AgdaBound{C}\AgdaSpace{}%
\AgdaBound{x}\AgdaSpace{}%
\AgdaBound{y}\AgdaSpace{}%
\AgdaSymbol{→}\AgdaSpace{}%
\AgdaBound{C}\AgdaSpace{}%
\AgdaBound{x′}\AgdaSpace{}%
\AgdaBound{y′}\<%
\\
\>[2]\AgdaFunction{tpt}\AgdaSpace{}%
\AgdaSymbol{\AgdaUnderscore{}}\AgdaSpace{}%
\AgdaSymbol{(}\AgdaInductiveConstructor{rfl}\AgdaSpace{}%
\AgdaSymbol{\AgdaUnderscore{})}\AgdaSpace{}%
\AgdaSymbol{(}\AgdaInductiveConstructor{rfl}\AgdaSpace{}%
\AgdaSymbol{\AgdaUnderscore{})}\AgdaSpace{}%
\AgdaBound{y}\AgdaSpace{}%
\AgdaSymbol{=}\AgdaSpace{}%
\AgdaBound{y}\<%
\\
\\[\AgdaEmptyExtraSkip]%
\>[2]\AgdaKeyword{module}\AgdaSpace{}%
\AgdaModule{\AgdaUnderscore{}}\AgdaSpace{}%
\AgdaSymbol{\{}\AgdaBound{l}\AgdaSpace{}%
\AgdaBound{m}\AgdaSymbol{\}\{}\AgdaBound{A}\AgdaSpace{}%
\AgdaSymbol{:}\AgdaSpace{}%
\AgdaPrimitiveType{Set}\AgdaSpace{}%
\AgdaBound{l}\AgdaSymbol{\}\{}\AgdaBound{B}\AgdaSpace{}%
\AgdaSymbol{:}\AgdaSpace{}%
\AgdaBound{A}\AgdaSpace{}%
\AgdaSymbol{→}\AgdaSpace{}%
\AgdaPrimitiveType{Set}\AgdaSpace{}%
\AgdaBound{m}\AgdaSymbol{\}}\AgdaSpace{}%
\AgdaKeyword{where}\<%
\\
\>[2][@{}l@{\AgdaIndent{0}}]%
\>[4]\AgdaFunction{fst}\AgdaSpace{}%
\AgdaSymbol{:}\AgdaSpace{}%
\AgdaDatatype{∑}\AgdaSpace{}%
\AgdaBound{A}\AgdaSpace{}%
\AgdaBound{B}\AgdaSpace{}%
\AgdaSymbol{→}\AgdaSpace{}%
\AgdaBound{A}\<%
\\
\>[4]\AgdaFunction{fst}\AgdaSpace{}%
\AgdaSymbol{(}\AgdaBound{x}\AgdaSpace{}%
\AgdaOperator{\AgdaInductiveConstructor{,}}\AgdaSpace{}%
\AgdaSymbol{\AgdaUnderscore{})}\AgdaSpace{}%
\AgdaSymbol{=}\AgdaSpace{}%
\AgdaBound{x}\<%
\\
\>[0]\<%
\\
\>[4]\AgdaFunction{snd}\AgdaSpace{}%
\AgdaSymbol{:}\AgdaSpace{}%
\AgdaSymbol{(}\AgdaBound{z}\AgdaSpace{}%
\AgdaSymbol{:}\AgdaSpace{}%
\AgdaDatatype{∑}\AgdaSpace{}%
\AgdaBound{A}\AgdaSpace{}%
\AgdaBound{B}\AgdaSymbol{)}\AgdaSpace{}%
\AgdaSymbol{→}\AgdaSpace{}%
\AgdaBound{B}\AgdaSpace{}%
\AgdaSymbol{(}\AgdaFunction{fst}\AgdaSpace{}%
\AgdaBound{z}\AgdaSymbol{)}\<%
\\
\>[4]\AgdaFunction{snd}\AgdaSpace{}%
\AgdaSymbol{(\AgdaUnderscore{}}\AgdaSpace{}%
\AgdaOperator{\AgdaInductiveConstructor{,}}\AgdaSpace{}%
\AgdaBound{y}\AgdaSymbol{)}\AgdaSpace{}%
\AgdaSymbol{=}\AgdaSpace{}%
\AgdaBound{y}\<%
\\
\>[0]\<%
\\
\>[4]\AgdaFunction{fpr}\AgdaSpace{}%
\AgdaSymbol{:}\AgdaSpace{}%
\AgdaSymbol{(}\AgdaBound{x}\AgdaSpace{}%
\AgdaSymbol{:}\AgdaSpace{}%
\AgdaBound{A}\AgdaSymbol{)(}\AgdaBound{y}\AgdaSpace{}%
\AgdaSymbol{:}\AgdaSpace{}%
\AgdaBound{B}\AgdaSpace{}%
\AgdaBound{x}\AgdaSymbol{)}\AgdaSpace{}%
\AgdaSymbol{→}\AgdaSpace{}%
\AgdaFunction{fst}\AgdaSpace{}%
\AgdaSymbol{(}\AgdaBound{x}\AgdaSpace{}%
\AgdaOperator{\AgdaInductiveConstructor{,}}\AgdaSpace{}%
\AgdaBound{y}\AgdaSymbol{)}\AgdaSpace{}%
\AgdaOperator{\AgdaFunction{≡}}\AgdaSpace{}%
\AgdaBound{x}\<%
\\
\>[4]\AgdaFunction{fpr}\AgdaSpace{}%
\AgdaBound{x}\AgdaSpace{}%
\AgdaSymbol{\AgdaUnderscore{}}\AgdaSpace{}%
\AgdaSymbol{=}\AgdaSpace{}%
\AgdaInductiveConstructor{rfl}\AgdaSpace{}%
\AgdaBound{x}\<%
\\
\>[0]\<%
\\
\>[4]\AgdaFunction{spr}\AgdaSpace{}%
\AgdaSymbol{:}\AgdaSpace{}%
\AgdaSymbol{(}\AgdaBound{x}\AgdaSpace{}%
\AgdaSymbol{:}\AgdaSpace{}%
\AgdaBound{A}\AgdaSymbol{)(}\AgdaBound{y}\AgdaSpace{}%
\AgdaSymbol{:}\AgdaSpace{}%
\AgdaBound{B}\AgdaSpace{}%
\AgdaBound{x}\AgdaSymbol{)}\AgdaSpace{}%
\AgdaSymbol{→}\AgdaSpace{}%
\AgdaFunction{snd}\AgdaSpace{}%
\AgdaSymbol{(}\AgdaBound{x}\AgdaSpace{}%
\AgdaOperator{\AgdaInductiveConstructor{,}}\AgdaSpace{}%
\AgdaBound{y}\AgdaSymbol{)}\AgdaSpace{}%
\AgdaOperator{\AgdaDatatype{≡≡}}\AgdaSpace{}%
\AgdaBound{y}\<%
\\
\>[4]\AgdaFunction{spr}\AgdaSpace{}%
\AgdaSymbol{\AgdaUnderscore{}}\AgdaSpace{}%
\AgdaBound{y}\AgdaSpace{}%
\AgdaSymbol{=}\AgdaSpace{}%
\AgdaInductiveConstructor{rfl}\AgdaSpace{}%
\AgdaBound{y}\<%
\\
\>[0]\<%
\\
\>[4]\AgdaFunction{eta}\AgdaSpace{}%
\AgdaSymbol{:}\AgdaSpace{}%
\AgdaSymbol{(}\AgdaBound{z}\AgdaSpace{}%
\AgdaSymbol{:}\AgdaSpace{}%
\AgdaDatatype{∑}\AgdaSpace{}%
\AgdaBound{A}\AgdaSpace{}%
\AgdaBound{B}\AgdaSymbol{)}\AgdaSpace{}%
\AgdaSymbol{→}\AgdaSpace{}%
\AgdaSymbol{(}\AgdaFunction{fst}\AgdaSpace{}%
\AgdaBound{z}\AgdaSpace{}%
\AgdaOperator{\AgdaInductiveConstructor{,}}\AgdaSpace{}%
\AgdaFunction{snd}\AgdaSpace{}%
\AgdaBound{z}\AgdaSymbol{)}\AgdaSpace{}%
\AgdaOperator{\AgdaFunction{≡}}\AgdaSpace{}%
\AgdaBound{z}\<%
\\
\>[4]\AgdaFunction{eta}\AgdaSpace{}%
\AgdaSymbol{(}\AgdaBound{x}\AgdaSpace{}%
\AgdaOperator{\AgdaInductiveConstructor{,}}\AgdaSpace{}%
\AgdaBound{y}\AgdaSymbol{)}\AgdaSpace{}%
\AgdaSymbol{=}\AgdaSpace{}%
\AgdaInductiveConstructor{rfl}\AgdaSpace{}%
\AgdaSymbol{(}\AgdaBound{x}\AgdaSpace{}%
\AgdaOperator{\AgdaInductiveConstructor{,}}\AgdaSpace{}%
\AgdaBound{y}\AgdaSymbol{)}\<%
\end{code}
Since we know from the previous section that these functions entail
Axiom~K, the above definitions have to use Agda's default
\texttt{--with-K} option to switch the existing implementation of
dependent pattern matching~\citep{CockxJ:eladcm} back to the original
version due to \cite{CoquandT:patmdt}, which is known to imply
Axiom~K~\citep{GoguenH:elidpm}. More precisely, it is only the matches
on the two occurrences of the pattern
$\AgdaInductiveConstructor{rfl}\,\_$ in the definition of $\afun{tpt}$
that involve an implicit use of Axiom~K (to discharge the unification
constraints $\aarg{A}\mathrel{≐}\aarg{A}$ and
$\aarg{B\,x}\mathrel{≐}\aarg{B\,x}$); all the other functions can be
defined without Axiom~K.

\section{Conclusion}
\label{sec:con}

This paper has investigated heterogeneous equality and produced a
simple collection of axioms for its typal form, in the spirit
of~\cite{CoquandT:equdtt}. The point of view is foundational. From a
practical perspective, the use of heterogeneous equality has much to
recommend it for formalizing mathematics in dependent type theory when
assuming uniqueness of identity proofs\footnote{Such is the approach
  of Lean~\citep{Lean} since version~3, for example.}; but that is
another story.


\appendix
\section*{Appendix: typal homogeneous equality without K}
\begin{code}[hide]%
\>[0]\AgdaKeyword{module}\AgdaSpace{}%
\AgdaModule{Appendix}\AgdaSpace{}%
\AgdaKeyword{where}\<%
\end{code}

In this appendix, for completeness sake we consider axioms in
dependent type theory without Axiom~K for
\emph{homogeneous} equality types 
\begin{code}%
\>[0][@{}l@{\AgdaIndent{1}}]%
\>[2]\AgdaKeyword{postulate}\<%
\\
\>[2][@{}l@{\AgdaIndent{0}}]%
\>[4]\AgdaOperator{\AgdaPostulate{\AgdaUnderscore{}≡\AgdaUnderscore{}}}\AgdaSpace{}%
\AgdaSymbol{:}\AgdaSpace{}%
\AgdaSymbol{∀\{}\AgdaBound{l}\AgdaSymbol{\}\{}\AgdaBound{A}\AgdaSpace{}%
\AgdaSymbol{:}\AgdaSpace{}%
\AgdaPrimitiveType{Set}\AgdaSpace{}%
\AgdaBound{l}\AgdaSymbol{\}}\AgdaSpace{}%
\AgdaSymbol{→}\AgdaSpace{}%
\AgdaBound{A}\AgdaSpace{}%
\AgdaSymbol{→}\AgdaSpace{}%
\AgdaBound{A}\AgdaSpace{}%
\AgdaSymbol{→}\AgdaSpace{}%
\AgdaPrimitiveType{Set}\AgdaSpace{}%
\AgdaBound{l}\<%
\end{code}
following~\cite{CoquandT:equdtt}. (Since without Axiom~K heterogeneous
equality is not very useful, we do not bother to consider axiomatizing
$\heq$ in that setting.)  One of the axioms makes use of dependent product
types. Although one could axiomatize those types as we did in the main
part of the paper, it is simpler to use an inductive defintion and
corresponding pair patterns:
\begin{code}%
\>[2]\AgdaKeyword{data}\AgdaSpace{}%
\AgdaDatatype{∑}\AgdaSpace{}%
\AgdaSymbol{\{}\AgdaBound{l}\AgdaSpace{}%
\AgdaBound{m}\AgdaSymbol{\}(}\AgdaBound{A}\AgdaSpace{}%
\AgdaSymbol{:}\AgdaSpace{}%
\AgdaPrimitiveType{Set}\AgdaSpace{}%
\AgdaBound{l}\AgdaSymbol{)(}\AgdaBound{B}\AgdaSpace{}%
\AgdaSymbol{:}\AgdaSpace{}%
\AgdaBound{A}\AgdaSpace{}%
\AgdaSymbol{→}\AgdaSpace{}%
\AgdaPrimitiveType{Set}\AgdaSpace{}%
\AgdaBound{m}\AgdaSymbol{)}\AgdaSpace{}%
\AgdaSymbol{:}\AgdaSpace{}%
\AgdaPrimitiveType{Set}\AgdaSpace{}%
\AgdaSymbol{(}\AgdaBound{l}\AgdaSpace{}%
\AgdaOperator{\AgdaPrimitive{⊔}}\AgdaSpace{}%
\AgdaBound{m}\AgdaSymbol{)}\AgdaSpace{}%
\AgdaKeyword{where}\<%
\\
\>[2][@{}l@{\AgdaIndent{0}}]%
\>[4]\AgdaOperator{\AgdaInductiveConstructor{\AgdaUnderscore{},\AgdaUnderscore{}}}\AgdaSpace{}%
\AgdaSymbol{:}\AgdaSpace{}%
\AgdaSymbol{(}\AgdaBound{x}\AgdaSpace{}%
\AgdaSymbol{:}\AgdaSpace{}%
\AgdaBound{A}\AgdaSymbol{)}\AgdaSpace{}%
\AgdaSymbol{→}\AgdaSpace{}%
\AgdaBound{B}\AgdaSpace{}%
\AgdaBound{x}\AgdaSpace{}%
\AgdaSymbol{→}\AgdaSpace{}%
\AgdaDatatype{∑}\AgdaSpace{}%
\AgdaBound{A}\AgdaSpace{}%
\AgdaBound{B}\<%
\\
\>[2]\AgdaComment{-- concrete syntax for ∑-types}\<%
\\
\>[2]\AgdaKeyword{syntax}\AgdaSpace{}%
\AgdaDatatype{∑}\AgdaSpace{}%
\AgdaBound{A}\AgdaSpace{}%
\AgdaSymbol{(λ}\AgdaSpace{}%
\AgdaBound{x}\AgdaSpace{}%
\AgdaSymbol{→}\AgdaSpace{}%
\AgdaBound{B}\AgdaSymbol{)}\AgdaSpace{}%
\AgdaSymbol{=}\AgdaSpace{}%
\AgdaDatatype{∑}\AgdaSpace{}%
\AgdaBound{x}\AgdaSpace{}%
\AgdaDatatype{∶}\AgdaSpace{}%
\AgdaBound{A}\AgdaSpace{}%
\AgdaDatatype{,}\AgdaSpace{}%
\AgdaBound{B}\<%
\\
\>[2]\AgdaComment{-- dependent product projections}\<%
\\
\>[2]\AgdaKeyword{module}\AgdaSpace{}%
\AgdaModule{\AgdaUnderscore{}}\AgdaSpace{}%
\AgdaSymbol{\{}\AgdaBound{l}\AgdaSpace{}%
\AgdaBound{m}\AgdaSymbol{\}\{}\AgdaBound{A}\AgdaSpace{}%
\AgdaSymbol{:}\AgdaSpace{}%
\AgdaPrimitiveType{Set}\AgdaSpace{}%
\AgdaBound{l}\AgdaSymbol{\}\{}\AgdaBound{B}\AgdaSpace{}%
\AgdaSymbol{:}\AgdaSpace{}%
\AgdaBound{A}\AgdaSpace{}%
\AgdaSymbol{→}\AgdaSpace{}%
\AgdaPrimitiveType{Set}\AgdaSpace{}%
\AgdaBound{m}\AgdaSymbol{\}}\AgdaSpace{}%
\AgdaKeyword{where}\<%
\\
\>[2][@{}l@{\AgdaIndent{0}}]%
\>[4]\AgdaFunction{fst}\AgdaSpace{}%
\AgdaSymbol{:}\AgdaSpace{}%
\AgdaDatatype{∑}\AgdaSpace{}%
\AgdaBound{A}\AgdaSpace{}%
\AgdaBound{B}\AgdaSpace{}%
\AgdaSymbol{→}\AgdaSpace{}%
\AgdaBound{A}\<%
\\
\>[4]\AgdaFunction{fst}\AgdaSpace{}%
\AgdaSymbol{(}\AgdaBound{x}\AgdaSpace{}%
\AgdaOperator{\AgdaInductiveConstructor{,}}\AgdaSpace{}%
\AgdaSymbol{\AgdaUnderscore{}}\AgdaSpace{}%
\AgdaSymbol{)}\AgdaSpace{}%
\AgdaSymbol{=}\AgdaSpace{}%
\AgdaBound{x}\<%
\\
\>[4]\AgdaFunction{snd}\AgdaSpace{}%
\AgdaSymbol{:}\AgdaSpace{}%
\AgdaSymbol{(}\AgdaBound{z}\AgdaSpace{}%
\AgdaSymbol{:}\AgdaSpace{}%
\AgdaDatatype{∑}\AgdaSpace{}%
\AgdaBound{A}\AgdaSpace{}%
\AgdaBound{B}\AgdaSymbol{)}\AgdaSpace{}%
\AgdaSymbol{→}\AgdaSpace{}%
\AgdaBound{B}\AgdaSpace{}%
\AgdaSymbol{(}\AgdaFunction{fst}\AgdaSpace{}%
\AgdaBound{z}\AgdaSymbol{)}\<%
\\
\>[4]\AgdaFunction{snd}\AgdaSpace{}%
\AgdaSymbol{(\AgdaUnderscore{}}\AgdaSpace{}%
\AgdaOperator{\AgdaInductiveConstructor{,}}\AgdaSpace{}%
\AgdaBound{y}\AgdaSymbol{)}\AgdaSpace{}%
\AgdaSymbol{=}\AgdaSpace{}%
\AgdaBound{y}\<%
\end{code}
The axioms for homogeneous equality are
\AgdaTarget{refl}
\AgdaTarget{cntr}
\AgdaTarget{sbst}
\begin{code}%
\>[2]\AgdaKeyword{postulate}\<%
\\
\>[2][@{}l@{\AgdaIndent{0}}]%
\>[4]\AgdaPostulate{refl}%
\>[10]\AgdaSymbol{:}\AgdaSpace{}%
\AgdaSymbol{∀\{}\AgdaBound{l}\AgdaSymbol{\}\{}\AgdaBound{A}\AgdaSpace{}%
\AgdaSymbol{:}\AgdaSpace{}%
\AgdaPrimitiveType{Set}\AgdaSpace{}%
\AgdaBound{l}\AgdaSymbol{\}}\AgdaSpace{}%
\AgdaSymbol{(}\AgdaBound{x}\AgdaSpace{}%
\AgdaSymbol{:}\AgdaSpace{}%
\AgdaBound{A}\AgdaSymbol{)}\AgdaSpace{}%
\AgdaSymbol{→}\AgdaSpace{}%
\AgdaBound{x}\AgdaSpace{}%
\AgdaOperator{\AgdaPostulate{≡}}\AgdaSpace{}%
\AgdaBound{x}\<%
\\
\>[4]\AgdaPostulate{cntr}%
\>[10]\AgdaSymbol{:}\AgdaSpace{}%
\AgdaSymbol{∀\{}\AgdaBound{l}\AgdaSymbol{\}\{}\AgdaBound{A}\AgdaSpace{}%
\AgdaSymbol{:}\AgdaSpace{}%
\AgdaPrimitiveType{Set}\AgdaSpace{}%
\AgdaBound{l}\AgdaSymbol{\}\{}\AgdaBound{x}\AgdaSpace{}%
\AgdaBound{y}\AgdaSpace{}%
\AgdaSymbol{:}\AgdaSpace{}%
\AgdaBound{A}\AgdaSymbol{\}(}\AgdaBound{e}\AgdaSpace{}%
\AgdaSymbol{:}\AgdaSpace{}%
\AgdaBound{x}\AgdaSpace{}%
\AgdaOperator{\AgdaPostulate{≡}}\AgdaSpace{}%
\AgdaBound{y}\AgdaSymbol{)}\AgdaSpace{}%
\AgdaSymbol{→}\AgdaSpace{}%
\AgdaSymbol{(}\AgdaBound{x}\AgdaSpace{}%
\AgdaOperator{\AgdaInductiveConstructor{,}}\AgdaSpace{}%
\AgdaPostulate{refl}\AgdaSpace{}%
\AgdaBound{x}\AgdaSymbol{)}\AgdaSpace{}%
\AgdaOperator{\AgdaPostulate{≡}}\AgdaSpace{}%
\AgdaSymbol{(}\AgdaBound{y}\AgdaSpace{}%
\AgdaOperator{\AgdaInductiveConstructor{,}}\AgdaSpace{}%
\AgdaBound{e}\AgdaSymbol{)}\<%
\\
\>[4]\AgdaPostulate{sbst}%
\>[10]\AgdaSymbol{:}\AgdaSpace{}%
\AgdaSymbol{∀\{}\AgdaBound{l}\AgdaSpace{}%
\AgdaBound{m}\AgdaSymbol{\}\{}\AgdaBound{A}\AgdaSpace{}%
\AgdaSymbol{:}\AgdaSpace{}%
\AgdaPrimitiveType{Set}\AgdaSpace{}%
\AgdaBound{l}\AgdaSymbol{\}(}\AgdaBound{B}\AgdaSpace{}%
\AgdaSymbol{:}\AgdaSpace{}%
\AgdaBound{A}\AgdaSpace{}%
\AgdaSymbol{→}\AgdaSpace{}%
\AgdaPrimitiveType{Set}\AgdaSpace{}%
\AgdaBound{m}\AgdaSymbol{)\{}\AgdaBound{x}\AgdaSpace{}%
\AgdaBound{x′}\AgdaSpace{}%
\AgdaSymbol{:}\AgdaSpace{}%
\AgdaBound{A}\AgdaSymbol{\}}\AgdaSpace{}%
\AgdaSymbol{→}\AgdaSpace{}%
\AgdaBound{x}\AgdaSpace{}%
\AgdaOperator{\AgdaPostulate{≡}}\AgdaSpace{}%
\AgdaBound{x′}\AgdaSpace{}%
\AgdaSymbol{→}\AgdaSpace{}%
\AgdaBound{B}\AgdaSpace{}%
\AgdaBound{x}\AgdaSpace{}%
\AgdaSymbol{→}\AgdaSpace{}%
\AgdaBound{B}\AgdaSpace{}%
\AgdaBound{x′}\<%
\end{code}
Coquand also considers a regularity axiom for $\afun{sbst}$
($\mathsf{ax}_3$ in \emph{loc.cit.}), but one can do without that by
using Peter Lumsdaine's trick to correct $\afun{sbst}$ to a version
$\afun{subst}$ for which there is a proof
$\afun{substIsRegular} : \forall\aarg{b} \fun
\afun{subst}\,(\afun{refl}\,\aarg{x})\,\aarg{b} \eq \aarg{b}$, as
follows. The proof begins as for Lemma~\ref{lem:verlt} by
considering functions that are injective modulo $\eq$:
\begin{code}%
\>[2]\AgdaFunction{Inj}\AgdaSpace{}%
\AgdaSymbol{:}\AgdaSpace{}%
\AgdaSymbol{∀\{}\AgdaBound{l}\AgdaSymbol{\}(}\AgdaBound{A}\AgdaSpace{}%
\AgdaBound{B}\AgdaSpace{}%
\AgdaSymbol{:}\AgdaSpace{}%
\AgdaPrimitiveType{Set}\AgdaSpace{}%
\AgdaBound{l}\AgdaSymbol{)}\AgdaSpace{}%
\AgdaSymbol{→}\AgdaSpace{}%
\AgdaPrimitiveType{Set}\AgdaSpace{}%
\AgdaBound{l}\<%
\\
\>[2]\AgdaFunction{Inj}\AgdaSpace{}%
\AgdaBound{A}\AgdaSpace{}%
\AgdaBound{B}\AgdaSpace{}%
\AgdaSymbol{=}\AgdaSpace{}%
\AgdaDatatype{∑}\AgdaSpace{}%
\AgdaBound{f}\AgdaSpace{}%
\AgdaDatatype{∶}\AgdaSpace{}%
\AgdaSymbol{(}\AgdaBound{A}\AgdaSpace{}%
\AgdaSymbol{→}\AgdaSpace{}%
\AgdaBound{B}\AgdaSymbol{)}\AgdaSpace{}%
\AgdaDatatype{,}\AgdaSpace{}%
\AgdaSymbol{∀\{}\AgdaBound{x}\AgdaSpace{}%
\AgdaBound{y}\AgdaSymbol{\}}\AgdaSpace{}%
\AgdaSymbol{→}\AgdaSpace{}%
\AgdaBound{f}\AgdaSpace{}%
\AgdaBound{x}\AgdaSpace{}%
\AgdaOperator{\AgdaPostulate{≡}}\AgdaSpace{}%
\AgdaBound{f}\AgdaSpace{}%
\AgdaBound{y}\AgdaSpace{}%
\AgdaSymbol{→}\AgdaSpace{}%
\AgdaBound{x}\AgdaSpace{}%
\AgdaOperator{\AgdaPostulate{≡}}\AgdaSpace{}%
\AgdaBound{y}\<%
\\
\>[0]\<%
\\
\>[2]\AgdaFunction{id}\AgdaSpace{}%
\AgdaSymbol{:}\AgdaSpace{}%
\AgdaSymbol{∀\{}\AgdaBound{l}\AgdaSymbol{\}\{}\AgdaBound{A}\AgdaSpace{}%
\AgdaSymbol{:}\AgdaSpace{}%
\AgdaPrimitiveType{Set}\AgdaSpace{}%
\AgdaBound{l}\AgdaSymbol{\}}\AgdaSpace{}%
\AgdaSymbol{→}\AgdaSpace{}%
\AgdaBound{A}\AgdaSpace{}%
\AgdaSymbol{→}\AgdaSpace{}%
\AgdaBound{A}\<%
\\
\>[2]\AgdaFunction{id}\AgdaSpace{}%
\AgdaBound{x}\AgdaSpace{}%
\AgdaSymbol{=}\AgdaSpace{}%
\AgdaBound{x}\<%
\\
\>[0]\<%
\\
\>[2]\AgdaFunction{idInj}\AgdaSpace{}%
\AgdaSymbol{:}\AgdaSpace{}%
\AgdaSymbol{∀\{}\AgdaBound{l}\AgdaSymbol{\}(}\AgdaBound{A}\AgdaSpace{}%
\AgdaSymbol{:}\AgdaSpace{}%
\AgdaPrimitiveType{Set}\AgdaSpace{}%
\AgdaBound{l}\AgdaSymbol{)}\AgdaSpace{}%
\AgdaSymbol{→}\AgdaSpace{}%
\AgdaFunction{Inj}\AgdaSpace{}%
\AgdaBound{A}\AgdaSpace{}%
\AgdaBound{A}\<%
\\
\>[2]\AgdaFunction{idInj}\AgdaSpace{}%
\AgdaSymbol{\AgdaUnderscore{}}\AgdaSpace{}%
\AgdaSymbol{=}\AgdaSpace{}%
\AgdaSymbol{(}\AgdaFunction{id}\AgdaSpace{}%
\AgdaOperator{\AgdaInductiveConstructor{,}}\AgdaSpace{}%
\AgdaFunction{id}\AgdaSymbol{)}\<%
\end{code}
But then to construct $\afun{subst}$ and $\afun{substIsRegular}$, one
has to work a bit harder than in the proof of the lemma, because of
the lack of uniqueness of identity proofs:
\AgdaTarget{subst}
\AgdaTarget{substIsRegular}
\begin{code}%
\>[2]\AgdaKeyword{module}\AgdaSpace{}%
\AgdaModule{\AgdaUnderscore{}}\AgdaSpace{}%
\AgdaSymbol{\{}\AgdaBound{l}\AgdaSpace{}%
\AgdaBound{m}\AgdaSymbol{\}\{}\AgdaBound{A}\AgdaSpace{}%
\AgdaSymbol{:}\AgdaSpace{}%
\AgdaPrimitiveType{Set}\AgdaSpace{}%
\AgdaBound{l}\AgdaSymbol{\}(}\AgdaBound{B}\AgdaSpace{}%
\AgdaSymbol{:}\AgdaSpace{}%
\AgdaBound{A}\AgdaSpace{}%
\AgdaSymbol{→}\AgdaSpace{}%
\AgdaPrimitiveType{Set}\AgdaSpace{}%
\AgdaBound{m}\AgdaSymbol{)\{}\AgdaBound{x}\AgdaSpace{}%
\AgdaSymbol{:}\AgdaSpace{}%
\AgdaBound{A}\AgdaSymbol{\}}\AgdaSpace{}%
\AgdaKeyword{where}\<%
\\
\>[2][@{}l@{\AgdaIndent{0}}]%
\>[4]\AgdaFunction{Inj₂}\AgdaSpace{}%
\AgdaSymbol{:}\AgdaSpace{}%
\AgdaSymbol{\{}\AgdaBound{y}\AgdaSpace{}%
\AgdaBound{z}\AgdaSpace{}%
\AgdaSymbol{:}\AgdaSpace{}%
\AgdaBound{A}\AgdaSymbol{\}}\AgdaSpace{}%
\AgdaSymbol{→}\AgdaSpace{}%
\AgdaBound{x}\AgdaSpace{}%
\AgdaOperator{\AgdaPostulate{≡}}\AgdaSpace{}%
\AgdaBound{y}\AgdaSpace{}%
\AgdaSymbol{→}\AgdaSpace{}%
\AgdaBound{x}\AgdaSpace{}%
\AgdaOperator{\AgdaPostulate{≡}}\AgdaSpace{}%
\AgdaBound{z}\AgdaSpace{}%
\AgdaSymbol{→}\AgdaSpace{}%
\AgdaFunction{Inj}\AgdaSpace{}%
\AgdaSymbol{(}\AgdaBound{B}\AgdaSpace{}%
\AgdaBound{y}\AgdaSymbol{)}\AgdaSpace{}%
\AgdaSymbol{(}\AgdaBound{B}\AgdaSpace{}%
\AgdaBound{z}\AgdaSymbol{)}\<%
\\
\>[4]\AgdaFunction{Inj₂}\AgdaSpace{}%
\AgdaSymbol{\{}\AgdaBound{y}\AgdaSymbol{\}}\AgdaSpace{}%
\AgdaBound{p}\AgdaSpace{}%
\AgdaBound{q}\AgdaSpace{}%
\AgdaSymbol{=}\<%
\\
\>[4][@{}l@{\AgdaIndent{0}}]%
\>[6]\AgdaPostulate{sbst}\AgdaSpace{}%
\AgdaSymbol{(λ}\AgdaSpace{}%
\AgdaBound{z′}\AgdaSpace{}%
\AgdaSymbol{→}\AgdaSpace{}%
\AgdaFunction{Inj}\AgdaSpace{}%
\AgdaSymbol{(}\AgdaBound{B}\AgdaSpace{}%
\AgdaBound{y}\AgdaSymbol{)}\AgdaSpace{}%
\AgdaSymbol{(}\AgdaBound{B}\AgdaSpace{}%
\AgdaBound{z′}\AgdaSymbol{))}\AgdaSpace{}%
\AgdaBound{q}\<%
\\
\>[6][@{}l@{\AgdaIndent{0}}]%
\>[7]\AgdaSymbol{(}\AgdaPostulate{sbst}\AgdaSpace{}%
\AgdaSymbol{(λ}\AgdaSpace{}%
\AgdaBound{y′}\AgdaSpace{}%
\AgdaSymbol{→}\AgdaSpace{}%
\AgdaFunction{Inj}\AgdaSpace{}%
\AgdaSymbol{(}\AgdaBound{B}\AgdaSpace{}%
\AgdaBound{y′}\AgdaSymbol{)}\AgdaSpace{}%
\AgdaSymbol{(}\AgdaBound{B}\AgdaSpace{}%
\AgdaBound{x}\AgdaSymbol{))}\AgdaSpace{}%
\AgdaBound{p}\AgdaSpace{}%
\AgdaSymbol{(}\AgdaFunction{idInj}\AgdaSpace{}%
\AgdaSymbol{(}\AgdaBound{B}\AgdaSpace{}%
\AgdaBound{x}\AgdaSymbol{)))}\<%
\\
\\[\AgdaEmptyExtraSkip]%
\>[0]\<%
\\
\>[4]\AgdaFunction{sbst₂}\AgdaSpace{}%
\AgdaSymbol{:}\AgdaSpace{}%
\AgdaSymbol{\{}\AgdaBound{y}\AgdaSpace{}%
\AgdaBound{z}\AgdaSpace{}%
\AgdaSymbol{:}\AgdaSpace{}%
\AgdaBound{A}\AgdaSymbol{\}}\AgdaSpace{}%
\AgdaSymbol{→}\AgdaSpace{}%
\AgdaBound{x}\AgdaSpace{}%
\AgdaOperator{\AgdaPostulate{≡}}\AgdaSpace{}%
\AgdaBound{y}\AgdaSpace{}%
\AgdaSymbol{→}\AgdaSpace{}%
\AgdaBound{x}\AgdaSpace{}%
\AgdaOperator{\AgdaPostulate{≡}}\AgdaSpace{}%
\AgdaBound{z}\AgdaSpace{}%
\AgdaSymbol{→}\AgdaSpace{}%
\AgdaBound{B}\AgdaSpace{}%
\AgdaBound{y}\AgdaSpace{}%
\AgdaSymbol{→}\AgdaSpace{}%
\AgdaBound{B}\AgdaSpace{}%
\AgdaBound{z}\<%
\\
\>[4]\AgdaFunction{sbst₂}\AgdaSpace{}%
\AgdaBound{p}\AgdaSpace{}%
\AgdaBound{q}\AgdaSpace{}%
\AgdaSymbol{=}\AgdaSpace{}%
\AgdaFunction{fst}\AgdaSpace{}%
\AgdaSymbol{(}\AgdaFunction{Inj₂}\AgdaSpace{}%
\AgdaBound{p}\AgdaSpace{}%
\AgdaBound{q}\AgdaSymbol{)}\<%
\\
\>[0]\<%
\\
\>[4]\AgdaFunction{C}\AgdaSpace{}%
\AgdaSymbol{:}\AgdaSpace{}%
\AgdaSymbol{\{}\AgdaBound{y}\AgdaSpace{}%
\AgdaSymbol{:}\AgdaSpace{}%
\AgdaBound{A}\AgdaSymbol{\}(}\AgdaBound{p}\AgdaSpace{}%
\AgdaSymbol{:}\AgdaSpace{}%
\AgdaBound{x}\AgdaSpace{}%
\AgdaOperator{\AgdaPostulate{≡}}\AgdaSpace{}%
\AgdaBound{y}\AgdaSymbol{)(}\AgdaBound{b}\AgdaSpace{}%
\AgdaSymbol{:}\AgdaSpace{}%
\AgdaBound{B}\AgdaSpace{}%
\AgdaBound{x}\AgdaSymbol{)}\AgdaSpace{}%
\AgdaSymbol{→}\AgdaSpace{}%
\AgdaDatatype{∑}\AgdaSpace{}%
\AgdaBound{c}\AgdaSpace{}%
\AgdaDatatype{∶}\AgdaSpace{}%
\AgdaBound{B}\AgdaSpace{}%
\AgdaBound{y}\AgdaSpace{}%
\AgdaDatatype{,}\AgdaSpace{}%
\AgdaSymbol{(}\AgdaFunction{sbst₂}\AgdaSpace{}%
\AgdaBound{p}\AgdaSpace{}%
\AgdaBound{p}\AgdaSpace{}%
\AgdaBound{c}\AgdaSpace{}%
\AgdaOperator{\AgdaPostulate{≡}}\AgdaSpace{}%
\AgdaFunction{sbst₂}\AgdaSpace{}%
\AgdaSymbol{(}\AgdaPostulate{refl}\AgdaSpace{}%
\AgdaBound{x}\AgdaSymbol{)}\AgdaSpace{}%
\AgdaBound{p}\AgdaSpace{}%
\AgdaBound{b}\AgdaSymbol{)}\<%
\\
\>[4]\AgdaFunction{C}%
\>[1894I]\AgdaBound{p}\AgdaSpace{}%
\AgdaBound{b}\AgdaSpace{}%
\AgdaSymbol{=}\AgdaSpace{}%
\AgdaPostulate{sbst}\AgdaSpace{}%
\AgdaFunction{C′}\AgdaSpace{}%
\AgdaSymbol{(}\AgdaPostulate{cntr}\AgdaSpace{}%
\AgdaBound{p}\AgdaSymbol{)}\AgdaSpace{}%
\AgdaSymbol{(}\AgdaBound{b}\AgdaSpace{}%
\AgdaOperator{\AgdaInductiveConstructor{,}}\AgdaSpace{}%
\AgdaPostulate{refl}\AgdaSpace{}%
\AgdaSymbol{\AgdaUnderscore{})}\<%
\\
\>[.][@{}l@{}]\<[1894I]%
\>[6]\AgdaKeyword{where}\<%
\\
\>[6]\AgdaFunction{C′}\AgdaSpace{}%
\AgdaSymbol{:}\AgdaSpace{}%
\AgdaDatatype{∑}\AgdaSpace{}%
\AgdaBound{y}\AgdaSpace{}%
\AgdaDatatype{∶}\AgdaSpace{}%
\AgdaBound{A}\AgdaSpace{}%
\AgdaDatatype{,}\AgdaSpace{}%
\AgdaSymbol{(}\AgdaBound{x}\AgdaSpace{}%
\AgdaOperator{\AgdaPostulate{≡}}\AgdaSpace{}%
\AgdaBound{y}\AgdaSymbol{)}\AgdaSpace{}%
\AgdaSymbol{→}\AgdaSpace{}%
\AgdaPrimitiveType{Set}\AgdaSpace{}%
\AgdaBound{m}\<%
\\
\>[6]\AgdaFunction{C′}\AgdaSpace{}%
\AgdaSymbol{(}\AgdaBound{y}\AgdaSpace{}%
\AgdaOperator{\AgdaInductiveConstructor{,}}\AgdaSpace{}%
\AgdaBound{p}\AgdaSymbol{)}\AgdaSpace{}%
\AgdaSymbol{=}\AgdaSpace{}%
\AgdaDatatype{∑}\AgdaSpace{}%
\AgdaBound{c}\AgdaSpace{}%
\AgdaDatatype{∶}\AgdaSpace{}%
\AgdaBound{B}\AgdaSpace{}%
\AgdaBound{y}\AgdaSpace{}%
\AgdaDatatype{,}\AgdaSpace{}%
\AgdaSymbol{(}\AgdaFunction{sbst₂}\AgdaSpace{}%
\AgdaBound{p}\AgdaSpace{}%
\AgdaBound{p}\AgdaSpace{}%
\AgdaBound{c}\AgdaSpace{}%
\AgdaOperator{\AgdaPostulate{≡}}\AgdaSpace{}%
\AgdaFunction{sbst₂}\AgdaSpace{}%
\AgdaSymbol{(}\AgdaPostulate{refl}\AgdaSpace{}%
\AgdaBound{x}\AgdaSymbol{)}\AgdaSpace{}%
\AgdaBound{p}\AgdaSpace{}%
\AgdaBound{b}\AgdaSymbol{)}\<%
\\
\\[\AgdaEmptyExtraSkip]%
\>[4]\AgdaFunction{subst}\AgdaSpace{}%
\AgdaSymbol{:}\AgdaSpace{}%
\AgdaSymbol{\{}\AgdaBound{y}\AgdaSpace{}%
\AgdaSymbol{:}\AgdaSpace{}%
\AgdaBound{A}\AgdaSymbol{\}}\AgdaSpace{}%
\AgdaSymbol{→}\AgdaSpace{}%
\AgdaBound{x}\AgdaSpace{}%
\AgdaOperator{\AgdaPostulate{≡}}\AgdaSpace{}%
\AgdaBound{y}\AgdaSpace{}%
\AgdaSymbol{→}\AgdaSpace{}%
\AgdaBound{B}\AgdaSpace{}%
\AgdaBound{x}\AgdaSpace{}%
\AgdaSymbol{→}\AgdaSpace{}%
\AgdaBound{B}\AgdaSpace{}%
\AgdaBound{y}\<%
\\
\>[4]\AgdaFunction{subst}\AgdaSpace{}%
\AgdaBound{p}\AgdaSpace{}%
\AgdaBound{b}\AgdaSpace{}%
\AgdaSymbol{=}\AgdaSpace{}%
\AgdaFunction{fst}\AgdaSpace{}%
\AgdaSymbol{(}\AgdaFunction{C}\AgdaSpace{}%
\AgdaBound{p}\AgdaSpace{}%
\AgdaBound{b}\AgdaSymbol{)}\<%
\\
\\[\AgdaEmptyExtraSkip]%
\>[4]\AgdaFunction{substIsRegular}\AgdaSpace{}%
\AgdaSymbol{:}\AgdaSpace{}%
\AgdaSymbol{(}\AgdaBound{b}\AgdaSpace{}%
\AgdaSymbol{:}\AgdaSpace{}%
\AgdaBound{B}\AgdaSpace{}%
\AgdaBound{x}\AgdaSymbol{)}\AgdaSpace{}%
\AgdaSymbol{→}\AgdaSpace{}%
\AgdaFunction{subst}\AgdaSpace{}%
\AgdaSymbol{(}\AgdaPostulate{refl}\AgdaSpace{}%
\AgdaBound{x}\AgdaSymbol{)}\AgdaSpace{}%
\AgdaBound{b}\AgdaSpace{}%
\AgdaOperator{\AgdaPostulate{≡}}\AgdaSpace{}%
\AgdaBound{b}\<%
\\
\>[4]\AgdaFunction{substIsRegular}\AgdaSpace{}%
\AgdaBound{b}\AgdaSpace{}%
\AgdaSymbol{=}\AgdaSpace{}%
\AgdaFunction{snd}\AgdaSpace{}%
\AgdaSymbol{(}\AgdaFunction{Inj₂}\AgdaSpace{}%
\AgdaSymbol{(}\AgdaPostulate{refl}\AgdaSpace{}%
\AgdaBound{x}\AgdaSymbol{)}\AgdaSpace{}%
\AgdaSymbol{(}\AgdaPostulate{refl}\AgdaSpace{}%
\AgdaBound{x}\AgdaSymbol{))}\AgdaSpace{}%
\AgdaSymbol{(}\AgdaFunction{snd}\AgdaSpace{}%
\AgdaSymbol{(}\AgdaFunction{C}\AgdaSpace{}%
\AgdaSymbol{(}\AgdaPostulate{refl}\AgdaSpace{}%
\AgdaBound{x}\AgdaSymbol{)}\AgdaSpace{}%
\AgdaBound{b}\AgdaSymbol{))}\<%
\end{code}

\end{document}